\documentclass[a4paper, headings=standardclasses]{scrartcl}
\pdfoutput=1
\usepackage[utf8]{inputenc}

\usepackage{mathpazo}
\usepackage{tgpagella}
\usepackage[T1]{fontenc}
\usepackage{textcomp}
\usepackage[varqu,varl]{zi4} %
\usepackage{amsthm}
\usepackage{amsmath}
\usepackage{amssymb}
\usepackage{fdsymbol}
\usepackage{csquotes}
\usepackage{booktabs}
\usepackage{array}
\usepackage{longtable}

\usepackage{mathpartir}

\usepackage[hidelinks]{hyperref}
\usepackage{thmtools} %

\usepackage{subcaption}

\usepackage{tikz}
\usetikzlibrary{calc}
\usetikzlibrary{math}
\usetikzlibrary{scopes}
\usetikzlibrary{arrows,arrows.meta}
\usetikzlibrary{shapes.symbols}
\usetikzlibrary{decorations.pathreplacing}
\usetikzlibrary{positioning}
\usetikzlibrary{shapes.multipart}

\newlength\algotmplength

\newcommand{\inlinealg}[1]{\begin{tikzpicture}[baseline] \node[anchor=base,inner sep=0.5pt, fill=black!5!white] (n) {\normalfont\bfseries\vrule height 1em depth 0.5em width 0pt #1}; \end{tikzpicture}\hskip0pt}

\makeatletter

\newcommand{\algorithm}[1]{%
  \begingroup%
  \leavevmode\nopagebreak[4]\vskip0pt%
  \setlength\parindent{0pt}%
  \setlength\parskip{0pt}%
  \setlength\algotmplength{\dimexpr\linewidth-1em}%
  \begin{tikzpicture}
    \node[inner sep=0.5em, fill=black!5!white] {\parbox{\algotmplength}{
      \bfseries\raggedright\strut\ignorespaces#1\strut}};
  \end{tikzpicture}%
  \endgroup%
}

\newcommand{\AW}{\vskip0.5em}
\newcommand{\AC}[1]{
  \strut%
  \begingroup%
  \leavevmode\vskip0pt%
  \setlength\algotmplength{\dimexpr\linewidth-1em}%
  \begin{tikzpicture}
    \node[inner sep=0pt] (X) {\parbox{\algotmplength}{\bfseries\raggedright\strut\ignorespaces#1\strut}};
    \draw[rounded corners=4pt,line width=0.66pt] ($(X.north west)+(-0.5em, 0.125em)$) -- ($(X.south west)+(-0.5em, -0.25em)$) -- +(1em, 0);
    \path (X.north west) +(0, 0.125em) (X.south west) +(-1em, -0.25em);
  \end{tikzpicture}%
  \endgroup%
  \ignorespaces%
}
\newcommand{\ACC}[1]{
  \strut%
  \begingroup%
  \leavevmode\vskip0pt%
  \setlength\algotmplength{\dimexpr\linewidth-1em}%
  \begin{tikzpicture}
    \node[inner sep=0pt] (X) {\parbox{\algotmplength}{\bfseries\raggedright\strut\ignorespaces#1\strut}};
    \draw[line width=0.66pt] ($(X.north west)+(-0.5em, 0.125em)$) -- ($(X.south west)+(-0.5em, -0.25em)$);
    \path (X.north west) +(0, 0.125em) (X.south west) +(-1em, -0.25em);
  \end{tikzpicture}%
  \endgroup%
  \vskip-0.25em
  \ignorespaces%
}
\newcommand{\ACX}[1]{
  \strut%
  \begingroup%
  \leavevmode\vskip0pt%
  \setlength\algotmplength{\dimexpr\linewidth-1em}%
  \begin{tikzpicture}
    \node[inner sep=0pt] (X) {\parbox{\algotmplength}{\bfseries\raggedright\strut\ignorespaces#1}};
    \draw[line width=0.66pt] ($(X.north west)+(-0.5em, 0.125em)$) -- ($(X.south west)+(-0.5em, 0)$);
    \path (X.north west) +(-1em, 0.125em);
  \end{tikzpicture}%
  \endgroup%
  \ignorespaces%
}

\newcommand{\AXC}[1]{
  \begingroup%
  \leavevmode%
  \setlength\algotmplength{\dimexpr\linewidth-1em}%
  \begin{tikzpicture}
    \node[inner sep=0pt] (X) {\parbox{\algotmplength}{\bfseries\raggedright\strut\ignorespaces#1\strut}};
    \draw[rounded corners=4pt,line width=0.66pt] ($(X.north west)+(-0.5em, 0)$) -- ($(X.south west)+(-0.5em, -0.25em)$) -- +(1em, 0);
    \path (X.north west) (X.south west) +(-1em, -0.25em);
  \end{tikzpicture}%
  \endgroup%
  \ignorespaces%
}
\newcommand{\proc}[1]{\textsc{\mdseries #1}}
\newcommand{\cmnt}[1]{{\mdseries\color{blue!80!black}$(\!\star$ #1 $\star\!)$}}

\makeatother

\usepackage{bm}

\usepackage[usestackEOL]{stackengine}

\setkomafont{disposition}{}
\RedeclareSectionCommands[font=\bfseries]{paragraph}

\usepackage[maxbibnames=99,maxcitenames=99,backref=true]{biblatex}

\usepackage{xpatch}

\xpatchbibdriver{misc}{\usebibmacro{doi+eprint+url}}{\usebibmacro{doi+eprint+url}\printfield{usera}}{}{\typeout{error: failed xpatch}}

\DeclareSourcemap{
    \maps[datatype=bibtex]{
      \map{
        \step[fieldsource=git]
        \step[fieldset=usera,origfieldval]
    }
  }
}

\DeclareFieldFormat{usera}{\textsc{git}: \url{#1}}

\AtEveryBibitem{\clearfield{month}}

\AtEveryBibitem{\clearfield{day}}
\AtEveryBibitem{\clearfield{issn}}

\DefineBibliographyStrings{english}{%
  backrefpage = {see p\adddot},%
  backrefpages = {see pp\adddot}%
}

\DeclareFieldFormat{pagerefformat}{\mkbibparens{{\color{blue!50!black}\mkbibemph{#1}}}}
\renewbibmacro*{pageref}{%
  \iflistundef{pageref}
    {}
    {\printtext[pagerefformat]{%
       \ifnumgreater{\value{pageref}}{1}
         {\bibstring{backrefpages}\ppspace}
         {\bibstring{backrefpage}\ppspace}%
       \printlist[pageref][-\value{listtotal}]{pageref}}}}

\newtheorem{theorem}{Theorem}

\newtheorem{lemma}[theorem]{Lemma}

\newtheorem{corollary}[theorem]{Corollary}

\theoremstyle{definition}
\newtheorem{definition}[theorem]{Definition}

\newtheoremstyle{break}
  {\topsep}{\topsep}%
  {}{}%
  {\bfseries}{}%
  { }{}%
\theoremstyle{break}

\newtheorem{algo}[theorem]{Algorithm}

\makeatletter
\newlength{\negph@wd}
\DeclareRobustCommand{\negphantom}[1]{%
  \ifmmode
    \mathpalette\negph@math{#1}%
  \else
    \negph@do{#1}%
  \fi
}
\newcommand{\negph@math}[2]{\negph@do{$\m@th#1#2$}}
\newcommand{\negph@do}[1]{%
  \settowidth{\negph@wd}{#1}%
  \hspace*{-\negph@wd}%
}
\makeatother

\DeclareMathOperator{\im}{im}
\newcommand{\cW}{\mathcal W}
\newcommand{\ang}[1]{{\langle #1 \rangle}}
\newcommand{\pwin}{\mathrel{\vec{\in}}}
\newcommand{\pwsubseteq}{\mathrel{\vec{\subseteq}}}

\usepackage{adjustbox}

\newcommand{\permsim}{\mathrel{\sim_{S_n}}}
\newcommand{\notpermsim}{\mathrel{\not\sim_{S_n}}}
\newcommand{\negsim}{\mathrel{\sim_{\ang{\mathbf b}}}}

\newcommand{\subs}{\mathrel{\raisebox{-2.75pt}{\ensurestackMath{\stackon[0.25pt]{\sim}{\adjustbox{trim=0pt 1.75pt 0pt 0pt, clip}{$\subseteq$}}}}}}

\newcommand{\ohot}[2]{\ang{#1{:}#2}}
\newcommand{\eps}{{\varepsilon}}

\title{An Answer to the Bose--Nelson Sorting Problem for 11 and 12 Channels}
\author{Jannis Harder\\\small\texttt{\href{mailto:me@jix.one}{me@jix.one}}}
\date{December 8, 2020\footnote{Minor revision on July 24, 2022}}

\usepackage[final]{microtype}

\hypersetup{final}
\begin{document}

\maketitle

\begin{abstract}
    We show that 11-channel sorting networks have at least 35 comparators and that 12-channel sorting networks have at least 39 comparators.
    This positively settles the optimality of the corresponding sorting networks given in The Art of Computer Programming vol.\ 3 and closes the two smallest open instances of the Bose--Nelson sorting problem.
    We obtain these bounds by generalizing a result of Van Voorhis from sorting networks to a more general class of comparator networks.
    From this we derive a dynamic programming algorithm that computes the optimal size for a sorting network with a given number of channels.
    From an execution of this algorithm we construct a certificate containing a derivation of the corresponding lower size bound, which we check using a program formally verified using the Isabelle/HOL proof assistant.
\end{abstract}
\section{Introduction} \label{sec:intro}

Sorting of finite sequences is a classical problem in algorithms.
It has many applications in theory and practice and is often used as a building block for various data structures and other algorithms.

General purpose sorting is usually done using comparison sorting which can only inspect the input elements by checking whether a pair of of elements is in a given order.
In many sorting algorithms a comparison is followed by an optional exchange of the compared elements fixing their resulting order, which is also called a compare-exchange operation.

For a fixed length $n$ of the input sequence, sorting networks are the subclass of these algorithms where the sequence of compare-exchange operations is oblivious, i.e fixed independently of the algorithm's input.
Here $n$ is also called the sorting network's number of channels.

Depending on the use case, one usually wants to minimize either the sorting network's size, which is the total number of comparisons, or the the depth, which is the largest number of comparisons performed on a single element.
Here we are only concerned with optimal size sorting networks.
For recent results on optimal depth sorting networks see \textcite{codishSortingNetworksEnd2019}.

Finding the optimal size $s(n)$ of a sorting network for a given input length $n$ is also known as the Bose--Nelson sorting problem, named after \textcite{boseSortingProblem1962} who gave the first sub-quadratic upper bound.

The asymptotic growth of $s(n)$ has been determined by \textcite{ajtaiLogSortingNetwork1983} who found a sorting network construction of size $O(n \log n)$, matching the general comparison sort lower bound of $\Omega(n \log n)$.
Their construction results in networks that are orders of magnitude larger than the smallest known sorting networks for all practical input sizes \cite{knuthArtComputerProgramming1998}. The same still holds for later constructions that improve the constant factor but maintain the optimal asymptotic growth, e.g.\ the construction by \textcite{goodrichZigzagSortSimple2014}.

In practice sorting networks are constructed using one of the $O(n\,(\log n)^2)$ methods by \textcite{batcherSortingNetworksTheir1968}.
Batcher's constructions recursively combine smaller sorting networks into larger ones, and thus every improvement for the best known sorting network size for small $n$ translates into improvements for larger networks constructed using these methods.

For small input lengths, the Bose--Nelson sorting problem was first answered for $n \le 8$ by Floyd and Knuth using a computer search for the $n = 7$ case \autocite{floydBoseNelsonSortingProblem1973}.

\textcite{vanvoorhisImprovedLowerBound1972} showed $s(n) \ge s(n - 1) + \lceil \log_2 n \rceil$, which improves upon the information theory bound of general comparison sorting. It is strict for all currently known $s(n)$ with $n$ even.

Finding the next value of $s(n)$ took over 40 years until \textcite{codishSortingNineInputs2016} computed $s(9) = 25$ and $s(10) = 29$.
For $n = 9$ they employed a computer search based on a new method they introduced, called \emph{generate-and-prune}.
For $n = 10$, given the value of $s(9)$, Van Voorhis's bound matched a known upper bound and so they also obtained the value of $s(10)$.

Every time a value of $s(n)$ has been determined so far, it matched the size of the then smallest known sorting network.
\textcite{knuthArtComputerProgramming1998} lists the following known upper bounds:

\begin{center}
\begin{tabular}{r@{}c@{}rrrrrrrrrrrr}
        $n$    & ${}={}$   & 1 & 2 & 3 & 4 & 5 &  6 &  7 &  8 &  9 & 10 & 11 & 12 \\
        $s(n)$ & ${}\le{}$ & 0 & 1 & 3 & 5 & 9 & 12 & 16 & 19 & 25 & 29 & 35 & 39 \\
\end{tabular}
\end{center}

The sorting networks establishing the upper bounds for $n = 11$ and $n = 12$ are attributed to G.~Shapiro and M.~W.~Green.

To compute $s(11)$ and $s(12)$, we introduce the notion of a partial sorting network in \autoref{sec:partialnets} and develop a theory of partial sorting networks and their optimal sizes.
For this we first restate established facts about sorting networks in terms of partial sorting networks.
Then, as our main contribution to the theory of sorting networks, we generalize Van Voorhis's bound to partial sorting networks.

To make effective use of our new result within a computer search, we need to ensure certain conditions, which we do by introducing well-behaved sequence sets in \autoref{sec:wellbehaved}.

In \autoref{sec:algo} we then develop a successive approximation dynamic programming algorithm for computing $s(n)$.
This algorithm uses our new result to prune large parts of the search space.
In \autoref{sec:implement} we describe some implementation details with further details in \autoref{sec:parallel}.

To be able to trust the result of our computations, in \autoref{sec:certificates} we extend our algorithm to produce a certificate for the lower bound it finds.
In \autoref{sec:verify}, we then formalize the Bose--Nelson sorting problem within the Isabelle/HOL proof assistant \autocite{nipkowIsabelleHOLProof2002}.
Based on that we implement a certificate checking algorithm with a machine checked formal proof of its correctness.

Finally in \autoref{sec:compute}, we use our parallel implementation to compute $s(11) = 35$, matching the known upper bound.
We verify the result for $s(11)$ using our formally verified certificate checker.
From that we also obtain $s(12) = 39$, again using Van Voorhis's lower bound and the known upper bounds.
\section{Preliminaries}

\paragraph{Sorting}

We are interested in sorting sequences of a fixed length $n$ which contain values of a totally ordered set $\Sigma$.
The set of all these sequences is denoted by $\Sigma^n$ and we write a sequence $x \in \Sigma^n$ and its elements as $x = (x_1, \ldots, x_n)$.

A sequence is \emph{sorted} if it is non-decreasing and the \emph{sorted sequence} $x^\mathbf s$ of a given sequence $x$ is the unique sorted sequence that is a permutation of $x$.
Computing $x^\mathbf s$ given $x$ is called \emph{sorting}.

We also call the indices $\{1, \ldots, n\}$ \emph{channels} and write $[n]$ for the set of all $n$ channels. The value $x_i$ of a sequence $x$ is also called the value of channel $i$ in $x$ or just channel $i$ in $x$.

\paragraph{Comparators and Exchanges}

To perform sorting we will use two operations, called \emph{exchanges} and \emph{comparators}.

The set of exchanges on $n$ channels, $E_n$, contains an exchange $(i,j)$ for each \emph{unordered} tuple $i, j \in [n]$ with $i \neq j$.
An exchange $(i,j)$ acts on a sequence $x \in \Sigma^n$, written $x^{(i,j)}$, by transposition of channels $i$ and $j$.

The set of comparators on $n$ channels, $C_n$, contains a comparator $[i,j]$ for each \emph{ordered} tuple $i, j \in [n]$ with $i \neq j$.
A comparator $[i,j]$ acts on on a sequence $x \in \Sigma^n$, written $x^{[i,j]}$, by comparing the order of $x_i$ and $x_j$.
If $x_i \le x_j$ the result is $x$, if $x_i \ge x_j$ the result is $x^{(i,j)}$, so that $(x^{[i,j]})_i \le (x^{[i,j]})_j$.

Note that, given a fixed sequence $x$, the actions of both, comparators and exchanges, are permutations.

Given a set of operations $A$, the set of all finite sequences containing only operations in $A$ is denoted by $A^*$.
We write sequences of operations by juxtaposition of its elements.
The action of a sequence of operations consists of each operation's action composed from \emph{left to right}, so that $x^{a_1a_2} = (x^{a_1})^{a_2}$ for $a_1, a_2 \in A$. The \emph{empty sequence}, denoted by $\eps$, has the identity as its action.

\paragraph{Comparator Networks and Sorting Networks}

A \emph{comparator network} $c$ on $n$ channels is a finite sequence of exchange and/or comparator operations $c \in (E_n \cup C_n)^*$.

A \emph{sorting network} $c$ on $n$ channels is a comparator network which has sorting as its action, i.e.\ $x^c = x^\mathbf s$ for all length-$n$ sequences $x$.

Two comparator networks are called \emph{equivalent} if they have the same action.
For example, all \emph{sorting networks} are equivalent by definition.

\begin{figure}
    \centering
    \begin{minipage}{.5\textwidth}
        \centering
\begin{tikzpicture}[
    line width=0.66pt,
    nscomp/.style={Circle[length=1.5mm,sep=-0.75mm]}-{Stealth[length=2.5mm, inset=0.5mm]},
    scomp/.style={Circle[length=1.5mm,sep=-0.75mm]}-{Circle[length=1.5mm,sep=-0.75mm]},
]

    \tikzmath { \xs = 4.5 / 10;}

    \newcommand{\cbend}[1]{-- +(0.2, 0) .. controls +(1, 0) and +(-1, 0) .. +(1.8, #1) -- ++(2, #1)}
    \newcommand{\cstep}[1]{-- ++(2 * #1, 0)}

    \newcommand{\nscomp}[1]{edge[shift={(1, 0)}, nscomp] +(0, #1)}

    \begin{scope}[
        y=-0.66cm,x=\xs cm,xshift=0,
    ]

        \draw (0, 0)
            \cstep{1}
            \cbend{1}
            \cstep{1}
            \cbend{1}
            \cstep{1}

            (0, 1)
            \cstep{1}
            \cbend{-1}
            \nscomp{2}
            \cstep{3}

            (0, 2)
            \nscomp{-2}
            \cstep{3}
            \cbend{-1}
            \nscomp{1}
            \cstep{1}
        ;

        \node at (1, 2.5) {$[3, 1]$};
        \node at (3, 2.5) {$(1, 2)$};
        \node at (5, 2.5) {$[1, 3]$};
        \node at (7, 2.5) {$(2, 3)$};
        \node at (9, 2.5) {$[2, 3]$};

    \end{scope}

\end{tikzpicture}%
        \captionof{figure}{Comparator network diagram}
        \label{fig:diagram}
    \end{minipage}%
    \begin{minipage}{.5\textwidth}
        \centering
\begin{tikzpicture}[
    line width=0.66pt,
    nscomp/.style={Circle[length=1.5mm,sep=-0.75mm]}-{Stealth[length=2.5mm, inset=0.5mm]},
    scomp/.style={Circle[length=1.5mm,sep=-0.75mm]}-{Circle[length=1.5mm,sep=-0.75mm]},
]

    \tikzmath { \xs = 4.5 / 10;}

    \newcommand{\cbend}[1]{-- +(0.2, 0) .. controls +(1, 0) and +(-1, 0) .. +(1.8, #1) -- ++(2, #1)}
    \newcommand{\cstep}[1]{-- ++(2 * #1, 0)}

    \newcommand{\nscomp}[1]{edge[shift={(1, 0)}, nscomp] +(0, #1)}
    \newcommand{\scomp}[1]{edge[shift={(1, 0)}, scomp] +(0, #1)}

    \begin{scope}[
        y=-0.66cm,x=\xs cm,xshift=0,
    ]

        \draw (0, 0)
            \scomp{2}
            \cstep{2}
            \scomp{1}
            \cstep{1}

            (0, 1)
            \cstep{1}
            \scomp{1}
            \cstep{2}

            (0, 2)
            \cstep{3}

        ;

        \node at (1, 2.5) {$[1, 3]$};
        \node at (3, 2.5) {$[2, 3]$};
        \node at (5, 2.5) {$[1, 2]$};

    \end{scope}

\end{tikzpicture}%
        \captionof{figure}{Standard form diagram}
        \label{fig:diagram_std}
    \end{minipage}%
\end{figure}

We can represent a comparator network graphically using a diagram (see \autoref{fig:diagram}).
The diagram consists of $n$ horizontal wires, one for each channel, numbered from top to bottom.
From left to right, for each exchange $(i, j)$ we cross the corresponding wires and for each comparator $[i, j]$ we draw a vertical arrow pointing from $i$ to $j$. The left side of the diagram represents the input to the sorting network and the right side represents the output.

\paragraph{The Bose--Nelson Sorting Problem}

We are interested in sorting as efficiently as possible.
Below, we will see that we can always avoid using exchanges, thus we define the \emph{size} of a comparator network $c$ as the number of comparator operations contained.
We denote the size of $c$ with $s(c)$.
The minimal size among all sorting network on $n$ is denoted by $s(n)$.
Determining $s(n)$ is known as the Bose--Nelson sorting problem \cite{floydBoseNelsonSortingProblem1973}.

\paragraph{Standard Forms and Permutations}

There are many different comparator networks that have the same action.
This is still the case if we restrict it to comparator networks having the same size and/or to sorting networks.

It is helpful to define a standard form of comparator networks that excludes some equivalent comparator networks.
We do this by defining a set of \emph{standard comparators}, $D_n = \{\, [i, j] \in C_n \mid i < j \,\}$.
When drawing a diagram for a network in which all comparators are standard, we usually draw them without arrow tips as they would all point downwards (see \autoref{fig:diagram_std}).

The symmetric group on the set of $n$ channels is denoted by $S_n$.
Given a permutation of channels $\sigma \in S_n$, we define its action on a sequence as rearranging the sequence's values according to the permutation such that $(x^\sigma)_{i^\sigma} = x_i$ for all channels $i \in [n]$.
Then every sequence of exchanges in $E_n^*$ has an action that is also the action of a permutation $\sigma \in S_n$.
The converse is also true: for every permutation $\sigma \in S_n$ we can find a sequence of exchanges $e_\sigma \in E_n^*$ such that $\sigma$ and $e_\sigma$ have the same action.
We fix one such sequence of exchanges $e_\sigma$ for every $\sigma \in S_n$.
We require that the identity permutation is mapped to the empty sequence of exchanges $\eps$, but allow arbitrary choices otherwise.
From here on we will identify $\sigma$ and $e_\sigma$, allowing us to write permutations as part of a comparator network.

A comparator network is in \emph{standard form} if it consists of a prefix $t \in D_n^*$ of standard comparators followed by a suffix that is a permutation $\sigma \in S_n$, i.e.\ $c$ is in \emph{standard form} if $c \in D_n^*S_n$.

A classic result by \textcite{floydBoseNelsonSortingProblem1973} states that every comparator network can be turned into an equivalent standard form comparator network of the same size. This can be done using a small set of rewrite rules that orient all comparators in the same direction and move all exchanges to the end (see \autoref{fig:stdrules}).
Additionally, if a \emph{sorting} network $c$ is in standard form, the permutation at the end is always the identity permutation, and thus we need no exchanges at all.

\begin{figure}
    \centering

\newcommand{\cbend}[1]{-- +(0.2, 0) .. controls +(1, 0) and +(-1, 0) .. +(1.8, #1) -- ++(2, #1)}
\newcommand{\cstep}[1]{-- ++(2 * #1, 0)}

\newcommand{\nscomp}[1]{edge[shift={(1, 0)}, nscomp] +(0, #1)}

\tikzmath { \xs = 4 / 11;}
\tikzmath { \yscale = -2 /3; }

\tikzset{
	every picture/.style={line width=0.66pt, font=\small},
	nscomp/.style={Circle[length=1.5mm,sep=-0.75mm]}-{Stealth[length=2.5mm, inset=0.5mm]},
	scomp/.style={Circle[length=1.5mm,sep=-0.75mm]}-{Circle[length=1.5mm,sep=-0.75mm]},
	slbl/.style={anchor=east},
	mlbl/.style={fill=white, inner sep=1pt, text opacity=1, fill opacity=0.85, font=\small},
	elbl/.style={anchor=west},
	baseline=(current bounding box.center),
}

\hfill
\begin{tikzpicture}[y=\yscale cm,x=\xs cm]
	\draw
	(0, 0)
	\cstep{1}

	(0, 1)
	\nscomp{-1}
	\cstep{1}

	(4, 0)
	\nscomp{1}
	\cstep{1}
	\cbend{1}

	(4, 1)
	\cstep{1}
	\cbend{-1}
	;

	\node at (3, 0.5) {$\Rightarrow$};

	\node at (1, 2) {$[j, i]$};
	\node at (5, 2) {$[i, j]$};
	\node at (7, 2) {$(i, j)$};
\end{tikzpicture}
\hfill
\begin{tikzpicture}[y=\yscale cm,x=\xs cm]
	\draw
	(0, 0)
	\cbend{1}
	\cstep{1}

	(0, 1)
	\cbend{-1}
	\nscomp{1}
	\cstep{1}

	(6, 0)
	\nscomp{1}
	\cstep{1}

	(6, 1)
	\cstep{1}
	;

	\node at (5, 0.5) {$\Rightarrow$};

	\node at (1, 2) {$(i, j)$};
	\node at (3, 2) {$[i, j]$};
	\node at (7, 2) {$[i, j]$};
\end{tikzpicture}
\hfill\strut

\vskip0.5em

\hfill
\begin{tikzpicture}[y=\yscale cm,x=\xs cm]
	\draw
	(0, 0)
	\cbend{1}
	\nscomp{1}
	\cstep{1}

	(0, 1)
	\cbend{-1}
	\cstep{1}

	(0, 2)
	\cstep{2}

	(6, 0)
	\nscomp{2}
	\cstep{1}
	\cbend{1}

	(6, 1)
	\cstep{1}
	\cbend{-1}

	(6, 2)
	\cstep{2}
	;

	\node at (5, 1) {$\Rightarrow$};

	\node at (1, 3) {$(i, j)$};
	\node at (3, 3) {$[j, k]$};
	\node at (7, 3) {$[i, k]$};
	\node at (9, 3) {$(i, j)$};
\end{tikzpicture}
\hfill
\begin{tikzpicture}[y=\yscale cm,x=\xs cm]
	\draw
	(0, 0)
	\cbend{1}
	\cstep{1}

	(0, 1)
	\cbend{-1}
	\cstep{1}

	(0, 2)
	\cstep{1}
	\nscomp{-1}
	\cstep{1}

	(6, 0)
	\cstep{1}
	\cbend{1}

	(6, 1)
	\cstep{1}
	\cbend{-1}

	(6, 2)
	\nscomp{-2}
	\cstep{2}
	;

	\node at (5, 1) {$\Rightarrow$};

	\node at (1, 3) {$(i, j)$};
	\node at (3, 3) {$[k, j]$};
	\node at (7, 3) {$[k, i]$};
	\node at (9, 3) {$(i, j)$};
\end{tikzpicture}
\hfill
\tikzmath { \yscale = -0.75 * 3  / 4; }
\begin{tikzpicture}[y=\yscale cm,x=\xs cm]
	\draw
	(0, 0)
	\cbend{1}
	\cstep{1}

	(0, 1)
	\cbend{-1}
	\cstep{1}

	(0, 2)
	\cstep{1}
	\nscomp{1}
	\cstep{1}

	(0, 3)
	\cstep{2}

	(6, 0)
	\cstep{1}
	\cbend{1}

	(6, 1)
	\cstep{1}
	\cbend{-1}

	(6, 2)
	\nscomp{1}
	\cstep{2}

	(6,3)
	\cstep{2}
	;

	\node at (5, 1.5) {$\Rightarrow$};

	\node at (1, 4) {$(i, j)$};
	\node at (3, 4) {$[k, l]$};
	\node at (7, 4) {$[k, l]$};
	\node at (9, 4) {$(i, j)$};
\end{tikzpicture}
\hfill\strut
    \caption{Rewrite rules transforming a comparator network into standard form}
    \label{fig:stdrules}
\end{figure}
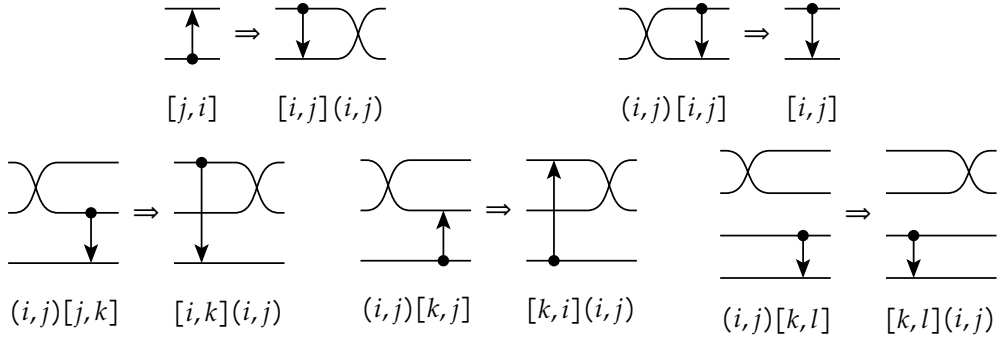

\paragraph{Comparator Networks as Circuits}

The diagrams for comparator networks hint at a different way to represent them.
We can view comparator networks as a circuit with $n$ ordered input ports and $n$ ordered output ports, built using only \emph{comparator gates}.
Each comparator gate has two input ports and two output ports, where the outputs are labeled \textbf{min} and \textbf{max} to distinguish them. The inputs are not labeled, as swapping the inputs of a comparator does not affect the output.
The \textbf{min} and \textbf{max} outputs take as values the minimum and the maximum of the two inputs, respectively.
Additionally we require that the circuit is acyclic and that every port is connected to exactly a single other port.
We will call such a circuit a \emph{comparator circuit}.

The diagram of a comparator network $c$, as described above, gives us such a circuit when we replace each arrow with a gate and consider the exchanges as part of the wiring between gates.
Applying a sequence $x$ to the inputs of that circuit results in $x^c$ on the outputs.

Given a comparator circuit, we can also construct a comparator network which computes the circuit's outputs as its action.
To do so, we perform a topological sort to get a linear order on the gates that respects their connectivity.
Then, beginning at the circuit's inputs, we assign channels to every wire such that the two wires connected to the outputs of a gate have the same two channels as the wires connected to the inputs.
For the outputs, we always assign the channel with the larger number to the \textbf{max} output.

\begin{figure}
    \centering

\begin{tikzpicture}[
    y=-0.5cm, x=0.5cm, line width=0.66pt, font=\small,
	baseline=(current bounding box.center),
]

    \newcommand{\gate}[1]{
        node[
            inner ysep=0pt,
            draw, rectangle split,
            rectangle split parts=5,
            minimum height=0.1pt,
            rectangle split draw splits=false,
            text width=1.2cm,
        ] (#1) {
            \nodepart{one}\vrule height 2pt depth 0pt width 0pt
            \nodepart{two}
            \textbf{in} \hfill \textbf{min}
            \nodepart{three}\vrule height 6pt depth 0pt width 0pt
            \nodepart{four}
            \textbf{in} \hfill \textbf{max}
            \nodepart{five}\vrule height 2pt depth 0pt width 0pt
        }
    }

    \newcommand{\ina}[1]{(#1.two west)}
    \newcommand{\inb}[1]{(#1.four west)}
    \newcommand{\outa}[1]{(#1.two east)}
    \newcommand{\outb}[1]{(#1.four east)}

    \path
        (0, 0) node (in1) {$1$}
        (0, 2) node (in2) {$2$}
        (0, 4) node (in3) {$3$}

        (3, 1.2)
        \gate{A}

        (7.25, 2.9)
        \gate{B}

        (11.5, 1)
        \gate{C}

        (15, 0) node (out1) {$1$}
        (15, 2) node (out2) {$2$}
        (15, 4) node (out3) {$3$}
    ;

    \newcommand{\curve}{.. controls +(10pt, 0) and +(-10pt, 0) ..}
    \newcommand{\wcurve}{.. controls +(30pt, 0) and +(-30pt, 0) ..}
    \newcommand{\lcurve}{-- +(2cm, 0) .. controls +(20pt, 0) and +(-20pt, 0) ..}

    \draw
        (in1.east) \curve \inb{A}
        (in2.east) \curve \ina{A}
        (in3.east) \wcurve \ina{B}
        \outa{A} \curve \inb{B}
        \outb{A} \wcurve \ina{C}
        \outb{B} \curve \inb{C}

        \outa{C} \curve (out1.west)
        \outb{C} \curve (out3.west)
        \outa{B} \lcurve (out2.west)
    ;

    \tikzset{lbl/.style={inner xsep=0.5pt, color=blue}}

    \draw  \ina{A} node[anchor=south east, lbl, inner ysep=1pt] {\scriptsize 2};
    \draw  \inb{A} node[anchor=north east, lbl, inner ysep=1pt] {\scriptsize 1};
    \draw \outa{A} node[anchor=south west, lbl, inner ysep=1pt] {\scriptsize 1};
    \draw \outb{A} node[anchor=north west, lbl, inner ysep=1pt] {\scriptsize 2};

    \draw  \ina{B} node[anchor=south east, lbl, inner ysep=1pt] {\scriptsize 3};
    \draw  \inb{B} node[anchor=north east, lbl, inner ysep=1pt] {\scriptsize 1};
    \draw \outa{B} node[anchor=south west, lbl, inner ysep=1.5pt] {\scriptsize 1};
    \draw \outb{B} node[anchor=north west, lbl, inner ysep=1pt] {\scriptsize 3};

    \draw  \ina{C} node[anchor=south east, lbl, inner ysep=1pt] {\scriptsize 2};
    \draw  \inb{C} node[anchor=north east, lbl, inner ysep=3pt] {\scriptsize 3};
    \draw \outa{C} node[anchor=south west, lbl, inner ysep=1pt] {\scriptsize 2};
    \draw \outb{C} node[anchor=north west, lbl, inner ysep=3pt] {\scriptsize 3};

    \draw (out1.west) node[anchor=south east, lbl, inner ysep=1pt] {\scriptsize 2};
    \draw (out2.west) node[anchor=south east, lbl, inner ysep=1pt] {\scriptsize 1};
    \draw (out3.west) node[anchor=south east, lbl, inner ysep=4pt] {\scriptsize 3};

\end{tikzpicture}
$\Rightarrow$
\begin{tikzpicture}[
    line width=0.66pt,
    nscomp/.style={Circle[length=1.5mm,sep=-0.75mm]}-{Stealth[length=2.5mm, inset=0.5mm]},
    scomp/.style={Circle[length=1.5mm,sep=-0.75mm]}-{Circle[length=1.5mm,sep=-0.75mm]},
	baseline=(current bounding box.center),
]

    \tikzmath { \xs = 4.5 / 10;}

    \newcommand{\cbend}[1]{-- +(0.2, 0) .. controls +(1, 0) and +(-1, 0) .. +(1.8, #1) -- ++(2, #1)}
    \newcommand{\cstep}[1]{-- ++(2 * #1, 0)}

    \newcommand{\nscomp}[1]{edge[shift={(1, 0)}, nscomp] +(0, #1)}
    \newcommand{\scomp}[1]{edge[shift={(1, 0)}, scomp] +(0, #1)}

    \begin{scope}[
        y=-0.66cm,x=\xs cm,xshift=0,
    ]

        \draw (0, 0)
            \scomp{1}
            \cstep{1}
            \scomp{2}
            \cstep{2}
            \cbend{1}

            (0, 1)
            \cstep{2}
            \scomp{1}
            \cstep{1}
            \cbend{-1}

            (0, 2)
            \cstep{4}

        ;

        \node at (1, 2.5) {$[1, 2]$};
        \node at (3, 2.5) {$[1, 3]$};
        \node at (5, 2.5) {$[2, 3]$};
        \node at (7, 2.5) {$(1,2)$};

    \end{scope}

\end{tikzpicture}%
    \caption{Construction of a comparator network from a comparator circuit}
    \label{fig:circuit}
\end{figure}
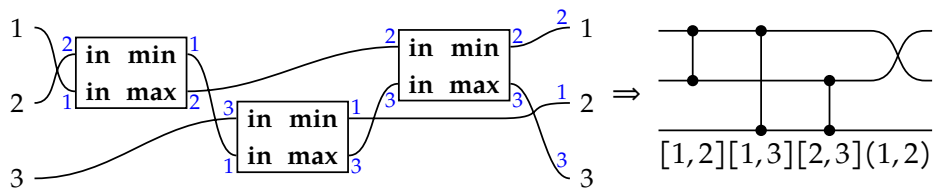

This gives us as comparator network $d \in D_n^*$, but as the channels assigned to the wires of the circuit's outputs might not match the order of the outputs, the action of $d$ produces a sequence that is a permutation of the circuit's outputs.
If $\sigma$ is the inverse of this permutation of the output, appending $\sigma$ gives a comparator network $d\sigma$ which has the circuit's outputs as its action.
See \autoref{fig:circuit} for an example.

As $d\sigma \in D_n^*S_n$ this network is in standard form.
The constructed comparator network is not unique as in general there are multiple possible ways to linearly order the gates.
We will define an arbitrary choice to be \emph{the} comparator network of a circuit.

Also note that the construction for either direction maintains the size of the comparator network, allowing us to freely switch between representations as suitable.

\paragraph{Sequence Sets}

A \emph{sequence set} of length $n$ is a subset of $\Sigma^n$.
For an operation $r$ that acts on sequences, we have a corresponding elementwise action on sequence sets, i.e.\ $X^r = \{\, x^r \mid x \in X \,\}$.
For a set of operations $R$, and a single sequence $x$ we apply each operation's action to the sequence and get $x^R = \{\, x^r \mid r \in R \,\}$.
In the following we will consider sequence sets to be the fundamental objects on which our operations act, and given a sequence set $X$ and a set of operations $R$ we will write $X^R$ for a set of sequence sets, each contained set being the result of an $r \in R$ acting on $X$, i.e.\ $\{\, \{\,x^r \mid x \in X\,\} \mid r \in R\,\}$.
Note that $X^R$ does not collapse into a set of sequences.
This differs from the more common convention which would collapse $X^R$ for sets $X$ and $R$ into a single set of sequences.

\paragraph{The Zero-One Principle}

So far we have not fixed the choice of $\Sigma$.
An important result for analyzing sorting networks is the zero-one principle which states that a comparator network is a sorting network if and only if it sorts all Boolean input sequences \cite{floydBoseNelsonSortingProblem1973,knuthArtComputerProgramming1998}.
We will denote the set of Boolean values $\{0, 1\}$ with $B$ and from here on, unless noted otherwise, will fix $\Sigma = B$.

\section{Partial Sorting Networks} \label{sec:partialnets}

To develop a search procedure for $s(n)$, we introduce the notion of a partial sorting network which has to sort only some input sequences.
We then derive several statements about the minimal size of partial sorting networks, on which we can base our search procedure.
These include well known properties of sorting networks, restated in terms of partial sorting networks, as well as novel results that are essential for our search procedure's performance.

\begin{definition}[Partial Sorting Networks]
	Given a set of input sequences $X \subseteq B^n$, a comparator network $c$ on $n$ channels is a \emph{partial sorting network} on $X$ if for each input sequence in $X$, the network's corresponding output is sorted, i.e. if $X^c = X^\mathbf s$.
	We use $s(X)$ to denote the minimal size required for a partial sorting network on $X$, which we will also call the \emph{minimal size for sorting} $X$.
\end{definition}

Next we derive several bounds on partial sorting networks which we will later use to build a recursive search procedure that can compute $s(X)$ for many $X$ including $B^n$.

\begin{lemma}[Minimal Size Inclusion Monotonicity] \label{fact:mono}
	For all sequence sets $X \subseteq Y \subseteq B^n$ the minimal sizes for sorting are bounded by $0 \le s(X) \le s(Y) \le s(B^n) = s(n)$.
\end{lemma}
\begin{proof}
	By the zero-one principle, a sorting network on $n$ channels is the same as a partial sorting network on $B^n$, thus $s(n) = s(B^n)$.
	As the number of comparators is always nonnegative we have $0 \le s(X)$.
	For all $X \subseteq Y \subseteq s(B^n)$ we have $s(X) \le s(Y) \le s(B^n)$, since every sorting network is a partial sorting network on $Y$ and in turn every partial sorting network on $Y$ also is a partial sorting network on $X$.
\end{proof}

\begin{definition}[Similarity] \label{def:sim}
	Given two sequence sets $X, Y \subseteq B^n$, we say $X$ and $Y$ are \emph{similar under permutation} or \emph{permutation-similar}, written $X \permsim Y$ if $X \in Y^{S_n}$, i.e. if $X$ and $Y$ are in the same orbit of the group action of $S_n$ on sequence sets.

	Given a sequence $x$, we denote its elementwise Boolean negation $x^\mathbf b$. Boolean negation and the identity operation form the two element cyclic group $\ang{\mathbf b}$.
	We say $X$ and $Y$ are \emph{similar under Boolean negation} or \emph{negation-similar}, written $X \negsim Y$ if $X \in Y^\ang{\mathbf b}$.

	Let $X \sim Y$ be the equivalence relation generated by both $\permsim$ and $\negsim$.
	We say $X$ is \emph{similar} to $Y$ when $X \sim Y$.
\end{definition}

As permutation of channels commutes with Boolean negation, given $X \sim Y$, we can always find a $T$ such that $X \permsim T \negsim Y$.

Next we show that the minimal size for sorting a sequence set is invariant under permutation and Boolean negation:

\begin{lemma}[Minimal Size Permutation Invariance] \label{fact:permute}
	All permutation-similar sequence sets $X \permsim Y$, have the same minimal size $s(X) = s(Y)$ for sorting.
\end{lemma}

\begin{proof}
	When $X \permsim Y$, there is a permutation $\sigma$ such that $X = Y^\sigma$. For all partial sorting networks $c$ on $Y$, we get the partial sorting network $\sigma c$ on $X$ having the same size as $c$, thus $s(X) \le s(Y)$. As $\permsim$ is an equivalence relation we also have $s(Y) \le s(X)$ by symmetry and therefore $s(X) = s(Y)$.
\end{proof}

\begin{lemma}[Minimal Size Negation Invariance] \label{fact:invert}
	All negation-similar sequence sets $X \negsim Y$, have the same minimal size $s(X) = s(Y)$ for sorting.
\end{lemma}

\begin{proof}
	We only have to handle the case when $X = Y^\mathbf b$ as otherwise $X = Y$.
	For a given partial sorting network $c$ on $Y$ we construct a comparator network $d$ by exchanging the outputs of every comparator.
	This construction gives us $z = y^c$ exactly if $z^\mathbf b = y^{\mathbf b d}$.
	As $z$ is sorted if and only if the reverse of $z^\mathbf b$ is sorted, we can construct a partial sorting network on $Y^\mathbf b = X$ by appending the permutation that reverses the order of channels.
	Let $\sigma$ be that permutation, then $d \sigma$ is the resulting network.
	This network has the same number of comparators as $c$ and we obtain $s(X) \le s(Y)$. As $\negsim$ is an equivalence relation we also have $s(Y) \le s(X)$ by symmetry and therefore $s(X) = s(Y)$.
\end{proof}

Combined we get the invariance under similarity.

\begin{corollary}[Minimal Size Similarity Invariance] \label{fact:sim}
	All similar sequence sets $X \sim Y$, have the same minimal size $s(X) = s(Y)$ for sorting.
\end{corollary}

Next we are going to combine set inclusion and similarity.

\begin{definition}[Subsumption] \label{def:subsume}
	Let $\subs$ be the preorder generated by inclusion ($\subseteq$) and similarity ($\sim$).
	When $X \subs Y$ for $X, Y \subseteq B^n$ we say that $X$ \emph{subsumes} $Y$.
\end{definition}

With this definition, invariance under similarity and monotonicity with respect to set inclusion combine to monotonicity with respect to subsumption.

\begin{corollary}[Minimal Size Subsumption Monotonicity] \label{fact:subsume}
	For all subsuming $X \subs Y$, the minimal sizes are related by
	$s(X) \le s(Y)$. \qed
\end{corollary}

As permutations and Boolean negation both act elementwise on a sequence set, $\sim$ and $\subseteq$ commute under composition of relations. Thus given $X \subs Y$ we can always find a $T$ such that $X \sim T \subseteq Y$.

A very similar notion of subsumption and a corresponding bound on the minimal sorting network size was used by \citeauthor*{codishSortingNineInputs2016} \cite{codishSortingNineInputs2016} to first compute $s(9)$.
Apart from the difference that we define it using partial sorting networks and that we include Boolean negation, they are identical.
A similar idea, used to compute the minimal depth of a sorting network, appears in \cite{bundalaOptimalSortingNetworks2014}.

Next we show that sorting a sufficiently large set of input sequences requires at least one comparator.

\begin{lemma}[Pigeon Hole Bound] \label{fact:nonempty}
	For all sequence sets $X \subseteq B^n$ with $\lvert X \rvert > n + 1$, we have a minimal size $s(X) \ge 1$ for sorting.
\end{lemma}

\begin{proof}
	Let $c \in E_n^*$ be an arbitrary comparator network without comparators.
	Application of $c$ is the same as application of a fixed permutation.
	Hence, for a set of input sequences $X$ and corresponding output sequences $Y$, we have $\lvert X \rvert=\lvert Y \rvert$.
	As there are only $|B^{n\mathbf s}| = n+1$ sorted Boolean sequences of length $n$ the network $c$ cannot sort a larger set of input sequences, including $X$. Thus a partial sorting network on such an $X$ requires at least one comparator.
\end{proof}

We can compute $s(X)$ exactly by considering all possible choices for the first comparator of a network, which must include the first comparator of each minimal size sorting network on $X$.

\begin{lemma}[Initial Comparator Recurrence] \label{fact:basicsucc}
	Given a sequence set $X \subseteq B^n$ with $s(X) \ge 1$, the minimal size for sorting $X$ is $s(X) = 1 + \min\,\{\, s(X^{[i,j]}) \mid [i, j] \in D_n \,\}$.
\end{lemma}

\begin{proof}
	Given a comparator $[i, j]$, for every partial sorting network $c$ on $X^{[i,j]}$ of size $m$, we have the partial sorting network $[i, j]c$ on $X$ of size $m + 1$.
	By choosing a minimal size $c$, we get the upper bound $s(X) \le 1 + \min\,\{\, s(X^{[i,j]}) \mid [i, j] \in D_n \,\}$.

	For a lower bound, let $d$ be a partial sorting network on $X$ in standard form of minimal size. As we require $s(X) \ge 1$, we know that $d$ has at least one comparator, and we have $d = [k, l]r$ for some comparator $[k, l]$ and a remaining suffix $r$.
	Now $r$ is a partial sorting network on $X^{[k,l]}$.
	As $d$ is of minimal size for sorting $X$, so is $r$ for sorting $X^{[k,l]}$ since all shorter partial sorting networks on $X^{[k,l]}$ could be extended to sort $X$ by prepending $[k, l]$.
	Thus $s(X) = 1 + s(X^{[k, l]})$.

	As $[k, l] \in D_n$, we get $s(X) = 1 + s(X^{[k, l]}) \ge 1 + \min\,\{\, s(X^{[i,j]}) \mid [i, j] \in D_n \,\}$ matching the lower bound.
\end{proof}

We can slightly improve this by excluding some possible choices for the first comparator.
For some sets $X$ with $s(X) \ge 1$ and some comparators $[k, l]$ we can observe that $X^{[k, l]} \permsim X$ will hold.
Here $[k, l]$ is called a \emph{redundant comparator}, which can be replaced by an exchange and thus is never part of a minimal size comparator network (see Exercise 5.3.4.15 of \cite{knuthArtComputerProgramming1998}, credited to R.~L.~Graham).
This fact is also used in the \enquote{Generate} step of \cite{codishSortingNineInputs2016}.

\begin{definition}[Successors]
	Given a sequence set $X \subseteq B^n$, the \emph{successors} of $X$ are the sequence sets $\{\, X^{[i,j]} \mid [i, j] \in D_n,\, X^{[i,j]} \notpermsim X \,\}$ and are denoted by $\mathcal U(X)$.
\end{definition}

\begin{lemma}[Successor Recurrence] \label{fact:succ}
	Given a sequence set $X \subseteq B^n$ with $s(X) \ge 1$, the minimal size is
	$s(X) = 1 + \min\,\{\, s(Y) \mid Y \in \mathcal U(X) \,\}$.
\end{lemma}

\begin{proof}
	When $X^{[k,l]} \permsim X$, there is a permutation $\sigma$ such that $X^{[k, l]} = X^\sigma$.
	In that case we have $s(X) = s(X^{[k,l]})$ by the permutation invariance (\autoref{fact:permute}).
	Effectively this means that $[k, l]$ could either be removed or replaced by the exchange $(k, l)$.

	At the same time, by the initial comparator recurrence (\autoref{fact:basicsucc}) we also have $s(X) = 1 + \min\,\{\, s(X^{[i,j]}) \mid [i, j] \in D_n \,\}$.
	This leads to a contradiction if $s(X^{[k,l]})$ were minimal among the $s(X^{[i,j]})$ with $[i, j] \in D_n$.

	Therefore, we can exclude these comparators from the minimum and get $s(X) = 1 + \min\,\{\, s(X^{[i, j]}) \mid [i, j] \in D_n,\, X^{[i, j]} \notpermsim X \,\} = 1 + \min\, \{\, s(Y) \mid Y \in \mathcal U(X) \,\}$.
\end{proof}

\subsection{Van Voorhis's Bound}

A very powerful lower bound for the size of sorting networks is due to Van Voorhis \cite{vanvoorhisImprovedLowerBound1972}.
It states that $s(n) \ge s(n-1) + \lceil \log_2 n \rceil$ for all $n \ge 1$.
This bound is strict for at least $n = 2, 3, 4, 6, 8, 10$, and as we will see also for $n = 12$.
These are more than half of the currently known values of $s(n)$.
Motivated by this we will later generalize it to partial sorting networks.

Van Voorhis's bound is obtained by showing that a sorting network on $n$ channels can be used to construct a sorting network on $n-1$ channels that is sufficiently smaller.
Such a network is constructed by tracing a path taken by a maximal input value and removing all comparators on that path.

\begin{definition}[Pruned Comparator Networks] \label{def:pruning}
	Given an $n$-channel comparator network $c$ and one of its channels $i \in [n]$,
	the \emph{pruned comparator network}, denoted by $c/i$, is a comparator network on $n - 1$ channels that we obtain by removing a path through the comparator network's circuit as described by the procedure below. We write $\delta(c, i)$ for the number $s(c) - s(c/i)$ of comparators removed while performing this procedure.

	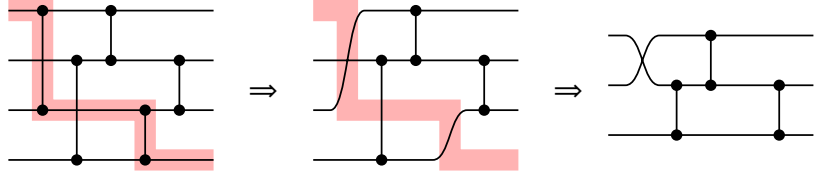
\begin{figure}
		\centering
\begin{tabular}{c}
\begin{tikzpicture}[
    line width=0.66pt,
    nscomp/.style={Circle[length=1.5mm,sep=-0.75mm]}-{Stealth[length=2.5mm, inset=0.5mm]},
    scomp/.style={Circle[length=1.5mm,sep=-0.75mm]}-{Circle[length=1.5mm,sep=-0.75mm]},
]

    \tikzmath { \xs = 4.5 / 20;}

    \newcommand{\cbend}[1]{-- +(0.2, 0) .. controls +(1, 0) and +(-1, 0) .. +(1.8, #1) -- ++(2, #1)}
    \newcommand{\cstep}[1]{-- ++(2 * #1, 0)}

    \newcommand{\nscomp}[1]{edge[shift={(1, 0)}, nscomp] +(0, #1)}
    \newcommand{\scomp}[1]{edge[shift={(1, 0)}, scomp] +(0, #1)}

    \begin{scope}[
        y=-0.66cm,x=\xs cm,xshift=0,
    ]

        \path[line width=8pt, draw=red!30!white]
            (0,0) --
            (2, 0) --
            (2, 2) --
            (8, 2) --
            (8, 3) --
            (12, 3)
            ;

        \draw (0, 0)
            \cstep{0.5}
            \scomp{2}
            \cstep{2}
            \scomp{1}
            \cstep{3}
            \cstep{0.5}

            (0, 1)
            \cstep{0.5}
            \cstep{1}
            \scomp{2}
            \cstep{3}
            \scomp{1}
            \cstep{1}
            \cstep{0.5}

            (0, 2)
            \cstep{0.5}
            \cstep{3}
            \scomp{1}
            \cstep{2}
            \cstep{0.5}

            (0, 3)
            \cstep{0.5}
            \cstep{5}
            \cstep{0.5}
        ;

    \end{scope}
\end{tikzpicture}
\end{tabular}
$\Rightarrow$
\begin{tabular}{c}
\begin{tikzpicture}[
    line width=0.66pt,
    nscomp/.style={Circle[length=1.5mm,sep=-0.75mm]}-{Stealth[length=2.5mm, inset=0.5mm]},
    scomp/.style={Circle[length=1.5mm,sep=-0.75mm]}-{Circle[length=1.5mm,sep=-0.75mm]},
]

    \tikzmath { \xs = 4.5 / 20;}

    \newcommand{\cbend}[1]{ .. controls +(1, 0) and +(-1, 0) ..  ++(2, #1)}
    \newcommand{\cstep}[1]{-- ++(2 * #1, 0)}

    \newcommand{\nscomp}[1]{edge[shift={(1, 0)}, nscomp] +(0, #1)}
    \newcommand{\scomp}[1]{edge[shift={(1, 0)}, scomp] +(0, #1)}

    \begin{scope}[
        y=-0.66cm,x=\xs cm,xshift=0,
    ]

        \path[line width=8pt, draw=red!30!white]
            (0,0) --
            (2, 0) --
            (2, 2) --
            (8, 2) --
            (8, 3) --
            (12, 3)
            ;

        \draw

            (0, 1)
            \cstep{0.5}
            \cstep{1}
            \scomp{2}
            \cstep{3}
            \scomp{1}
            \cstep{1}
            \cstep{0.5}

            (0, 2)
            \cstep{0.5}
            \cbend{-2}
            \cstep{1}
            \scomp{1}
            \cstep{3}
            \cstep{0.5}

            (0, 3)
            \cstep{0.5}
            \cstep{3}
            \cbend{-1}
            \cstep{1}
            \cstep{0.5}
        ;

    \end{scope}
\end{tikzpicture}
\end{tabular}
$\Rightarrow$
\begin{tabular}{c}
\begin{tikzpicture}[
    line width=0.66pt,
    nscomp/.style={Circle[length=1.5mm,sep=-0.75mm]}-{Stealth[length=2.5mm, inset=0.5mm]},
    scomp/.style={Circle[length=1.5mm,sep=-0.75mm]}-{Circle[length=1.5mm,sep=-0.75mm]},
]

    \tikzmath { \xs = 4.5 / 20;}

    \newcommand{\cbend}[1]{ .. controls +(1, 0) and +(-1, 0) ..  ++(2, #1)}
    \newcommand{\cstep}[1]{-- ++(2 * #1, 0)}

    \newcommand{\nscomp}[1]{edge[shift={(1, 0)}, nscomp] +(0, #1)}
    \newcommand{\scomp}[1]{edge[shift={(1, 0)}, scomp] +(0, #1)}

    \begin{scope}[
        y=-0.66cm,x=\xs cm,xshift=0,
    ]

        \draw

            (0, 0)
            \cstep{0.5}
            \cbend{1}
            \scomp{1}
            \cstep{3}
            \scomp{1}
            \cstep{1}
            \cstep{0.5}

            (0, 1)
            \cstep{0.5}
            \cbend{-1}
            \cstep{1}
            \scomp{1}
            \cstep{3}
            \cstep{0.5}

            (0, 2)
            \cstep{0.5}
            \cstep{5}
            \cstep{0.5}
        ;

    \end{scope}
\end{tikzpicture}%
\end{tabular}
		\captionof{figure}{Pruning the first channel of a sorting network}
		\label{fig:prune}
	\end{figure}

	To obtain $c/i$, we start with the circuit for $c$ and remove all wires along the path starting at input $i$ following the $\mathbf{max}$ output of every comparator encountered and ending at an output $j$ of the circuit.
	After removing these wires, the comparator gates on the path are left with a single incoming and a single outgoing wire each.
	We remove each such gate together with its adjacent wires and insert a shortcutting wire connected the two ports from which we just removed wires.
	Finally, we remove the input $i$ and output $j$ which are now isolated, keeping the order of the remaining inputs and outputs. See \autoref{fig:prune} for an example.
\end{definition}

To analyze a pruned comparator network we will need the following two definitions.

\begin{definition}[Pruned Sequences]
	For a sequence $x \in \Sigma^n$ and a channel $i \in [n]$, the \emph{pruned sequence} $x/i$ is obtained by removing the channel $i$ from $x$, resulting in $(x_1, \ldots, x_{i-1}, x_{i+1}, \ldots, x_n) \in \Sigma^{n - 1}$.
\end{definition}

\begin{definition}[One-Hot Sequences]
	For a channel $i \in [n]$ the \emph{one-hot} sequence $\ohot in$ is the unique sequence of length $n$ for which $\ohot in / i = (0, \ldots, 0)$ and $\ohot in_i = 1$.
\end{definition}

First we will show that the resulting smaller network on $n - 1$ channels is indeed a sorting network.

\begin{lemma}[Pruned Sorting Networks] \label{fact:prunedsorts}
	Given an $n$-channel sorting network $c$, the comparator network $c/i$ is a sorting network for all channels $i \in [n]$.
\end{lemma}

\begin{proof}
	The path removed when pruning $c$ is exactly the path obtained by applying $c$ to $\ohot in$ and following the value $1$ through the network's circuit.
	Since $c$ is a sorting network, it will produce the output $\ohot nn$, the only sorted one-hot sequence, and the pruned path will end in channel $n$.

	Furthermore, as all operations in a comparator network are monotone, the pruned path
	will have a constant value of $1$ for all input sequences $x$ with $x_i = 1$.
	This means that each comparator gate on the path will pass the value of the input that is not on the path to the output that is also not on the path, matching the new connections in the pruned network.

	From this we can see that given an input sequence $x \in B^n$ with $x_i = 1$ we get the same result whether we apply $c$ and then prune the output sequence, or whether we prune the input sequence and then apply $c/i$, so we have $x^c/i = (x/i)^{c/i}$.
	As every input sequence $y \in B^{n - 1}$ has such a corresponding $x$ with $x_i = 1$ and $x/i = y$ and every $x^c/i$ is sorted, we can conclude that every $y^{c/i}$ is also sorted.
	Thus $c/i$ is a sorting network.
\end{proof}

\begin{lemma}[Van Voorhis \cite{vanvoorhisImprovedLowerBound1972}] \label{fact:vanvoorhis}
	For all $n \ge 1$ the size of an $n$-channel sorting network is bounded by $s(n) \ge s(n-1) + \lceil \log_2 n \rceil$.
\end{lemma}

\begin{proof}
	Let $c$ be an arbitrary minimal size $n$-channel sorting network.
	For every channel $i \in [n]$, we obtain a pruned comparator network $c/i$.
	As each $c/i$ is a sorting network with $n - 1$ channels, it has size $s(c/i) \ge s(n - 1)$.
	Since we removed $\delta(c,i) = s(c) - s(c/i)$ comparators when constructing $c/i$, we get that $s(n) \ge s(n - 1) + \delta(c,i)$ for every channel $i$.
	Combining these bounds we get $s(n) \ge s(n - 1) + \max\,\{\, \delta(c,i) \mid i \in [n] \,\}$.

	As we do not know the values of $\delta(c,i)$, we will use a lower bound for $\max\,\{\, \delta(c,i) \mid i \in [n] \,\}$ independent of $c$.

	\begin{figure}
		\centering
\begin{tabular}{c}
\begin{tikzpicture}[
    line width=0.66pt,
    nscomp/.style={Circle[length=1.5mm,sep=-0.75mm]}-{Stealth[length=2.5mm, inset=0.5mm]},
    scomp/.style={Circle[length=1.5mm,sep=-0.75mm]}-{Circle[length=1.5mm,sep=-0.75mm]},
]

    \tikzmath { \xs = 4.5 / 20;}

    \newcommand{\cbend}[1]{-- +(0.2, 0) .. controls +(1, 0) and +(-1, 0) .. +(1.8, #1) -- ++(2, #1)}
    \newcommand{\cstep}[1]{-- ++(2 * #1, 0)}

    \newcommand{\nscomp}[1]{edge[shift={(1, 0)}, nscomp] +(0, #1)}
    \newcommand{\scomp}[1]{edge[shift={(1, 0)}, scomp] +(0, #1)}
    \newcommand{\xcomp}[1]{edge[shift={(1, 0)}, scomp, draw=blue] +(0, #1)}

    \begin{scope}[
        y=-0.66cm,x=\xs cm,xshift=0,
    ]

        \path[line width=8pt, draw=red!30!white]
            (0, 0) --
            (2, 0) --
            (2, 3) --
            (16, 3) --
            (16, 4)

            (0, 1) --
            (4, 1) --
            (4, 4)

            (0, 2) --
            (12, 2) --
            (12, 4)

            (0, 3) --
            (16, 3)

            (0, 4) --
            (20, 4)
        ;

        \path[line width=1pt, xshift=4.5pt, draw=white]
            (2, 0.5) -- (2, 2.5)
            (4, 1.5) -- (4, 3.5)
            (12, 2.5) -- (12, 3.5)
        ;

        \path[line width=1pt, xshift=-4.5pt, draw=white]
            (2, 0.5) -- (2, 2.5)
            (4, 1.5) -- (4, 3.5)
            (12, 2.5) -- (12, 3.5)
        ;

        \draw (0, 0)
            \cstep{0.5}
            \scomp{3}
            \cstep{2}
            \scomp{2}
            \cstep{2}
            \scomp{1}
            \cstep{5}
            \cstep{0.5}

            (0, 1)
            \cstep{0.5}
            \cstep{1}
            \scomp{3}
            \cstep{2}
            \scomp{2}
            \cstep{3}
            \scomp{1}
            \cstep{3}
            \cstep{0.5}

            (0, 2)
            \cstep{0.5}
            \cstep{5}
            \scomp{2}
            \cstep{3}
            \scomp{1}
            \cstep{1}
            \cstep{0.5}

            (0, 3)
            \cstep{0.5}
            \cstep{7}
            \scomp{1}
            \cstep{2}
            \cstep{0.5}

            (0, 4)
            \cstep{0.5}
            \cstep{9}
            \cstep{0.5}
        ;

        \path
            (-1, 0) node {\strut$1$}
            (-1, 1) node {\strut$2$}
            (-1, 2) node {\strut$3$}
            (-1, 3) node {\strut$4$}
            (-1, 4) node {\strut$5$}

            (2, -0.5) node {\strut$a$}
            (4,  0.5) node {\strut$b$}
            (6, -0.5) node {\strut$c$}
            (8,  0.5) node {\strut$d$}
            (12, 1.5) node {\strut$e$}
            (16, 2.4) node {\strut$f$}
        ;

    \end{scope}
\end{tikzpicture}
\end{tabular}
$\Rightarrow$
\begin{tabular}{c}
\begin{tikzpicture}[inner sep=5pt]
    \node (f) {$f$};
        \node[below of=f, xshift=-1cm] (d) {$d$};
            \node[below of=d] (a) {$a$};
                \node[below of=a, xshift=-0.5cm] (i1) {$1$};
                \node[below of=a, xshift=0.5cm] (i4) {$4$};
        \node[below of=f, xshift=1.25cm] (e) {$e$};
            \node[below of=e, xshift=-0.75cm] (c) {$c$};
                \node[below of=c] (i3) {$3$};
            \node[below of=e, xshift=0.75cm] (b) {$b$};
                \node[below of=b, xshift=-0.5cm] (i2) {$2$};
                \node[below of=b, xshift=0.5cm] (i5) {$5$};

    \draw
        (f)
            edge (d)
            edge (e)
        (d)
            edge (a)
        (e)
            edge (c)
            edge (b)
        (a)
            edge (i1)
            edge (i4)
        (c)
            edge (i3)
        (b)
            edge (i2)
            edge (i5)
    ;
\end{tikzpicture}
\end{tabular}
		\captionof{figure}{Binary tree of pruned paths}
		\label{fig:prune_tree}
	\end{figure}
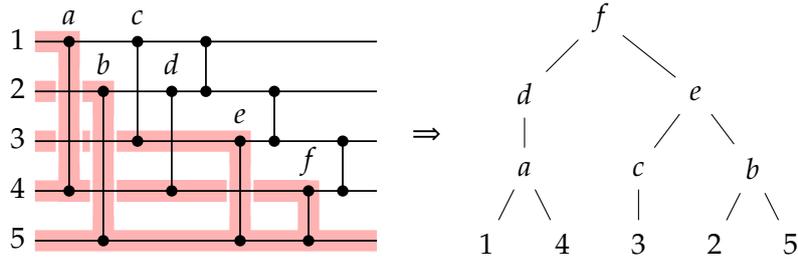

	For this we observe that $\delta(c,i)$ is the number of comparators on the pruned path when constructing $c/i$ (\autoref{def:pruning}).
	In the proof of \autoref{fact:prunedsorts} we saw that every pruned path ends in the output channel $n$.
	Thus, if we take the union of the pruned paths for all $i$ and remove the common output $n$, we obtain a binary tree rooted in the comparator gate connected to output $n$, where the leaves are all inputs of $c$ and all inner vertices are comparators gates (see \autoref{fig:prune_tree}).

	In this tree, the leaf for every input $i$ has a depth of $\delta(c,i)$, as the number of comparators gates on the pruned path is exactly the length of the path from the leaf to the root.
	Now $ \max\,\{\, \delta(c,i) \mid i \in [n] \,\}$ is the maximal depth among the leafs, which is, by definition, the height of the tree.
	As a binary tree with $n$ leaves has a height of at least $\lceil \log_2 n \rceil$ we now have a suitable lower bound.

	Replacing $ \max\,\{\, \delta(c,i) \mid i \in [n] \,\}$ with this lower bound we obtain $s(n) \ge s(n-1) + \lceil \log_2 n \rceil$.
\end{proof}

\subsection{Generalizing Van Voorhis's bound}

Given how effective Van Voorhis's bound is, we wish to obtain a similar bound for partial sorting networks, in the hope that this enables use of such a bound as part of a search procedure.

For sorting networks, we used the fact that we can sort all sequences of length $n - 1$ using a sorting network on $n$ channels by inserting a value of $1$ at a fixed channel $i$ of the input.
For a partial sorting network, this only works for some sequences of length $n - 1$.
To show that the pruned network matches the original network on the remaining inputs and outputs, we used that $\ohot in$ will be sorted by all sorting networks and then used the monotonicity of the sorting network operations.
Again for a partial sorting network this is not necessarily the case for every channel $i$.

We work around this by defining a suitable set of sequences to be sorted by the pruned partial sorting network and by restricting the set of inputs that we prune.

\begin{definition}[Pruned Sequence Sets] \label{def:prunedset}
	For a given sequence set $X$, the \emph{pruned sequence set} $\{\, x/i \mid x \in X,\, x_i = 1 \,\}$ is denoted by $X/i$.
\end{definition}

\begin{definition}[Prunable Channels] \label{def:prunable}
	A channel $i \in [n]$ is a \emph{prunable channel} of $X \subseteq B^n$ if $X$ contains the one-hot sequence $\ohot in$. The set of prunable channels of $X$ is denoted by $p(X)$.
\end{definition}

\begin{lemma}[Pruned Partial Sorting Networks] \label{fact:prunedpartialsorts}
	Given a partial sorting network $c$ on a sequence set $X \subseteq B^n$, for every prunable channel $i \in p(X)$, the comparator network $c/i$ is a partial sorting network on $X/i$.
\end{lemma}

\begin{proof}
	As $\ohot in \in X$ for $i \in p(X)$, we know that applying $c$ to $\ohot in \in X$ results in the output sequence $\ohot nn$ and thus the pruned path starting at channel $i$ will end in channel $n$.
	Now the same argument as used in the proof for sorting networks (\autoref{fact:prunedsorts}), restricted to the input sequences $X/i$, applies.
\end{proof}

We now try to follow the steps in the proof of Van Voorhis's bound starting with these definitions for partial sorting networks.
Let $c$ be a minimal size partial sorting network on $X$. For every prunable channel $i$, we get the bound $s(X) \ge s(X/i) + \delta(c,i)$. If we combine all these bounds, we get
\[
	s(X) \ge \max\, \{\, s(X/i) + \delta(c,i) \mid i \in p(X) \,\}.
\]

At this point we cannot continue as we did for Van Voorhis's bound, as both $s(X/i)$ and $\delta(c,i)$ depend on the channel $i$.
This means where before, we used the height of the binary tree as a lower bound for the maximal $\delta(c,i)$, we now need to find a lower bound for the maximal $s(X/i) + \delta(c,i)$ over all prunable $i \in p(X)$, independent of the given comparator network $c$.

We will again construct the binary tree of all paths that we can prune.
This time, as we have to consider the different values of $s(X/i)$ for different $i \in p(X)$, so we label the leaf for input $i$ with the value $s(X/i)$.
The depth of that leaf is again $\delta(c,i)$.

Now let the \emph{cost} of a leaf be the sum of its value and its depth and the cost of a tree the maximal cost among all its leaves.
By this definition, $\max\, \{\, s(X/i) + \delta(c,i) \mid i \in p(X) \,\}$ is the cost of the binary tree of pruned paths.

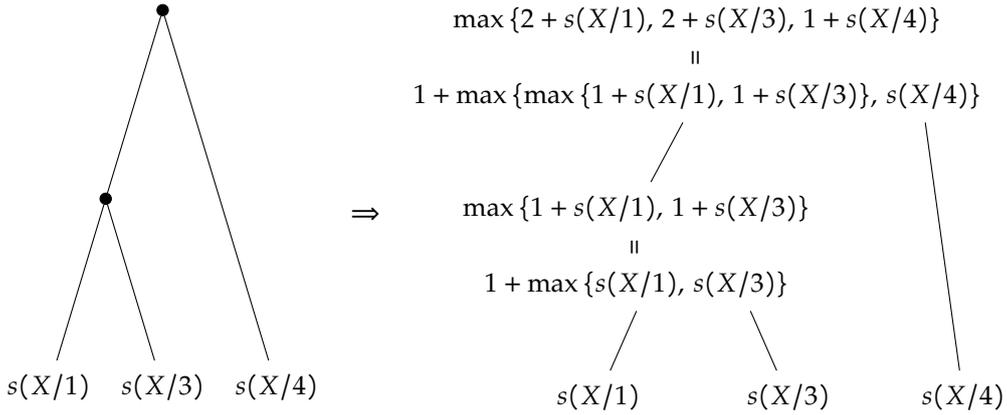
\begin{figure}
	\centering
\begin{tabular}{c}
\begin{tikzpicture}[inner sep=5pt, node distance=2.5cm, font=\small]
    \node[circle, draw=black, fill=black, inner sep=1.5pt] (root) {};
        \node[circle, draw=black, below of=root, fill=black, inner sep=1.5pt, xshift=-0.75cm] (left) {};
            \node[below of=left, xshift=-0.75cm] (left_left) {$s(X/1)$};
            \node[below of=left, xshift=0.75cm] (left_right) {$s(X/3)$};
        \node[below of=root, node distance=5cm, xshift=1.5cm] (right) {$s(X/4)$};

    \draw
        (root)
            edge (left)
            edge (right)
        (left)
            edge (left_left)
            edge (left_right)
    ;
\end{tikzpicture}
\end{tabular}
$\!\Rightarrow\!$
\begin{tabular}{c}
\begin{tikzpicture}[inner sep=5pt, font=\small]
    \node[node distance=0.5cm] (root_mid) {$\max\, \{2 + s(X/1),\, 2 + s(X/3),\, 1 + s(X/4)\}$};
    \node[below of=root_mid, node distance=0.5cm, rotate=90] (root_eq2) {$=$};
    \node[below of=root_eq2, node distance=0.5cm] (root) {$1 + \max\, \{\max\,\{1 + s(X/1),\, 1 + s(X/3)\},\, s(X/4)\}$};

    \node[below of=root, node distance=1.5cm, xshift=-0.8cm] (left_top) {$\max\,\{1 + s(X/1),\, 1 + s(X/3)\}$};
    \node[below of=left_top, node distance=0.5cm, rotate=90] (left_eq) {$=$};
    \node[below of=left_eq, node distance=0.5cm] (left) {$1 + \max\, \{s(X/1),\, s(X/3)\}$};

    \node[below of=left, node distance=1.5cm, xshift=-0.5cm] (left_left) {$s(X/1)$};
    \node[below of=left, node distance=1.5cm, xshift=2cm] (left_right) {$s(X/3)$};

    \node[below of=root, node distance=4cm, xshift=3.5cm] (right) {$s(X/4)$};

    \draw
    (root)
        edge (left_top)
    ($(root.south)+(3,0)$)
        edge (right)
    ($(left.south)+(0,0)$)
        edge (left_left)
    ($(left.south)+(1.5,0)$)
        edge (left_right)
    ;

\end{tikzpicture}
\end{tabular}
	\captionof{figure}{Recursive cost computation}
	\label{fig:cost_tree}
\end{figure}

To find a lower bound for this cost, we will first describe how to compute it using the structure of the tree.
We do this using a process that labels the root of every subtree with the subtree's cost.
Then, the cost of each subtree can be computed recursively.
For a singleton tree there is nothing to be done, as the leaf at the root has depth $0$ so the existing leaf label matches the cost.
For a non-singleton tree, in each immediate subtree, all leaves have a depth that is reduced by one compared to their depth in the complete tree.
The same is true for the cost as the leaf values stay the same.
Thus the cost of a tree is one more than the maximum cost of its immediate subtrees.
\autoref{fig:cost_tree} shows an example of a small tree labeled in this way.

It is possible that some inner nodes have only one child.
As we prefer to work with full binary trees, we remove all such nodes by replacing them with their child and relabel the tree using the same procedure.
This can only reduce the tree's cost as it reduces the depth of some leaves while the leaf values stay the same.
In particular, all lower bounds for the resulting full binary tree's cost are also valid for the initial tree.

To find a lower bound for the cost of such a tree, we define a class of trees that includes these.

\begin{definition}[$\star$-Trees]
	Given a commutative binary operation $\star$ and a multiset $L$, a \emph{$\star$-tree on $L$} is a
	a labeled rooted full binary tree with leaf labels $L$ and with each inner node label $x$ given by $a \star b$ where $a$ and $b$ are that inner node's children's labels.
\end{definition}

For the binary operation $1 + \max\,\{a, b\}$ and the multiset $M = \{\, s(X/i) \mid i \in p(X)\,\}_\#$, the tree we constructed above is a $1{+}$max-tree on $M$.

To get a lower bound for $\max\, \{\, s(X/i) + \delta(c,i) \mid i \in p(X) \,\}$, independent of $c$, we use the minimal cost of all $1{+}$max-trees on $M$.
This is analogous to using the minimal height of all binary trees with $n$ leaves when deriving Van Voorhis's bound.

Computing this minimal cost can be done by a generalization of Huffman's algorithm~\cite{knuthHuffmanAlgorithmAlgebra1982}.
This generalization is parameterized over a binary operation $\star$.
The algorithm will build a $\star$-tree on a given multiset $L$ and for a suitable binary operation that $\star$-tree will have the smallest root label among all $\star$-trees on $L$.

For this to hold, the binary operation needs to form a \emph{Huffman algebra}:

\begin{definition}[Huffman Algebra \cite{knuthHuffmanAlgorithmAlgebra1982}] \label{def:huffmanalgebra}
	A \emph{Huffman algebra} $(A, \le, \star)$ is a set $A$ with a total order $\le$ and a binary operation $\star$ satisfying the following axioms:
	\begin{equation*}
		\begin{aligned}
			a &\le a \star b \\
			a \star b &= b \star a \\
			(a \star b) \star (c \star d) &= (a \star c) \star (b \star d) \\
			&\negphantom{(a \star b) \star c}\begin{aligned}
			a \star c           & \le b \star c           & \text{ if } a \le b \\
			(a \star b) \star c & \le a \star (b \star c) & \text{ if } a \le c
		\end{aligned}
		\end{aligned}
	\end{equation*}
\end{definition}

Huffman's algorithm builds a $\star$-tree with a minimal root label by starting with a forest of isolated nodes labeled with the values in $L$, which will be the leaves of the final $\star$-tree.
As long as the forest is not connected, it combines the two $\star$-trees having the smallest values at their root by making them children of a new root node where ties are broken arbitrarily.

Note that if we are only interested in the value of the root, we do not need to completely store the $\star$-trees of the forest. The algorithm only inspects the labels of the roots, so it is enough to store these labels as a multiset, resulting in the following algorithm:

\begin{algo}[Huffman's Algorithm \cite{knuthHuffmanAlgorithmAlgebra1982}]\label{alg:huffman}
	\algorithm {
		function $H_\star\,L\colon$ \AC{
			$V \gets L$

			while $|V|_\# > 1\colon$ \AC{
				$a \gets \min_\le V;\quad b \gets \min_\le (V - \{a\}_\#)$

				$V \gets V -  \{a, b\}_\# + \{a \star b\}_\#$
			}

			return $v \in V$
		}
	}
\end{algo}

\begin{lemma}[Huffman's Algorithm on a Huffman Algebra \cite{knuthHuffmanAlgorithmAlgebra1982}] \label{fact:huffmanalg}
	Given a Huffman algebra $(A, \star, \le)$ and a nonempty finite multiset $L$ with members in $A$, $H_\star\,L$ returns the minimal root value among all $\star$-trees on $L$.
\end{lemma}

The following corollary allows us to use Huffman's algorithm when given only lower bounds on the leaf values:

\begin{corollary}[Monotonicity of Huffman's Algorithm]\label{fact:huffmanmono}
	Given a Huffman algebra $(A, \star, \le)$ and $a_1, \ldots, a_k \in A$, the result of Huffman's algorithm $H_\star\,\{a_1, \ldots, a_k\}_\#$ is monotone in each of $a_1, \ldots, a_k$.
\end{corollary}

\begin{proof}
	The $\star$ operation is monotone in each argument, thus the root value of every $\star$-tree on $\{a_1, \ldots, a_k\}_\#$ is monotone in each leaf value $a_1, \ldots, a_k$. Huffman's algorithm computes the minimum of all these root values.
\end{proof}

Equipped with Huffman's algorithm, we can state and prove our main theorem about partial sorting networks:

\begin{theorem}[The Huffman Bound] \label{fact:huffmanbound}
	Given a sequence set $X \subseteq B^n$ with prunable channels $p(X) \ne \emptyset$, the size for sorting $X$ is bounded by $s(X) \ge H_{{1{+}\mathrm{max}}}\,\{\, s(X/i) \mid i \in p(X) \,\}_\#$.
\end{theorem}

\begin{proof}
	We saw that the minimal root value among all the $1{+}$max-trees having leaf labels ${\{\, s(X/i) \mid i \in p(X) \,\}_\#}$ is a lower bound for $s(X)$.
	We can verify that $(\mathbb N_0, \le, 1{+}\mathrm{max})$ is a Huffman algebra by distributing $+$ over $\max$ in the Huffman algebra axioms.
	Thus Huffman's algorithm computes this minimal root value $H_{{1{+}\mathrm{max}}}\,\{\, s(X/i) \mid i \in p(X)\,\}_\#$.
\end{proof}

Note that while the Huffman bound is not affected by a permutation of channels, it is not symmetric with respect to Boolean negation.
The similar sequence sets $X$ and $X^\mathbf b$ can have different Huffman bounds even though $s(X) = s(X^\mathbf b)$.
Therefore it is often desirable to compute it for both sequence sets to get a stronger combined bound.

\section{Well-Behaved Sequence Sets} \label{sec:wellbehaved}

The Huffman bound is only applicable for sequence sets that have at least one prunable channel.
There are many sequence sets without prunable channels.
Even if a sequence set has prunable channels, it might contain sequences which are among the excluded sequences for every allowed pruning operation.
This in turn means that the Huffman bound, which is computed solely from those pruned sequence sets, is ineffective.
It cannot distinguish such a sequence set from the sequence set with those sequences removed.

If we start with $B^n$ and apply any given comparator network, we can observe that we will not encounter such a sequence set.
Thus we might hope, that the Huffman bound is effective for all sequence sets we will encounter.
As soon as we make use of the Huffman bound and start pruning sequence sets in addition to applying comparators, we will quickly find examples of sequence sets containing sequences that are excluded by each possible pruning operation.

As an example, consider the following sequence set:
\begin{align*}
    X &= (B^6)^{[1, 2][3, 4][1, 3][2, 4][5, 6][1, 5][4, 6]}/6/5 \\
      &= \{(0, 0, 0, 0), (0, 0, 0, 1), (0, 0, 1, 1), (0, 1, 0, 1), (0, 1, 1, 0), (0, 1, 1, 1), (1, 1, 1, 1) \}
\end{align*}

The only prunable channel of $X$ is channel $4$, for which pruning will exclude the sequence $(0, 1, 1, 0)$.

This turned out to be a practical issue during the development of an algorithm that uses the Huffman bound.

To ensure that we can make efficient use of the Huffman bound, we will define a suitable family of well-behaved sequence sets that excludes such sequence sets, and, apart from pruning, is closed under various operations we want to perform.
We will also define well-behaved pruned sequence sets which are maximum well-behaved subsets of the corresponding pruned sequence sets.
This allows us to apply a slightly weaker variant of the Huffman bound, while staying within the well-behaved sequence sets.

We could work around this issue completely by not using Boolean sequence sets, but instead using sets containing sequences that are permutations of $(1, \ldots, n)$.
Then for every such permutation sequence set $X$ and every channel $i$ for which there is a sequence $x \in X$ with $x_i = n$, the channel $i$ would be prunable.
The corresponding pruning operation $X/i$ would select exactly these sequences.
As every sequence in $X$ has the value $n$ for some channel, this ensures that every sequence will be kept by one of the possible pruning operations and thus also that every nonempty permutation sequence set has prunable channels.

Only superficial changes would be required to adapt the Huffman bound and the other statements about partial sorting networks from Boolean sequence sets to permutation sequence sets.
As permutation sequence sets are significantly larger than Boolean sequence sets, though, this is not practical for our purposes.

Instead we can lift the desired properties from permutation sequence sets to Boolean sequence sets, using an approach inspired by the zero-one principle.
For this, we will represents a single permutation of $(1, \ldots, n)$ using multiple Boolean sequences.

\begin{definition}[Threshold Sets] \label{def:thresholdsets}
    Let $R_n$ be the set of all sequences that are permutations of $(1, \ldots, n)$, i.e.\ $R_n = (1, \ldots, n)^{S_n}$.
    Consider the monotone functions from $[n]$ to the Booleans.
    These are the $n + 1$ \emph{threshold functions} $t_k(x) = [x \ge k]$ for all $1 \le k \le n + 1$ where $[P] = 1$ if the property $P$ holds and $[P] = 0$ otherwise.

    Given a sequence $x = (1, \ldots, n)^\sigma \in R_n$, we then define the \emph{threshold set} $T(x)$ as the set obtained by pointwise application of each threshold function to $x$, i.e. $T(x) = \{\,x^{t_k} \mid 1 \le k \le n + 1\,\}$. The set of all length-$n$ threshold sets is denoted by $\mathcal T_n$.

    Note that the threshold set $T((1, \ldots, n))$ consists of exactly all sorted Boolean sequences $(B^n)^\mathbf s$ and that pointwise function application and permutation of channels commute, so we can alternatively define threshold sets by $T((1, \ldots, n)^\sigma) = T((1, \ldots, n))^\sigma = (B^n)^{\mathbf s\sigma}$ for all permutations $\sigma$.
\end{definition}

\begin{definition}[Well-Behaved Sequence Sets]
    A Boolean sequence set $X \subseteq B^n$ is \emph{well-behaved} if it is a union of threshold sets, each in $\mathcal T_n$.
    The family of well-behaved sequence sets of length $n$ is denoted by $\cW_n$ and the family of well-behaved sequence sets of all lengths by $\cW$.
\end{definition}

\begin{lemma}[Well-Behaved Sequence Preservation]
    Given a well-behaved $X \in \cW_n$, for every sequence $x \in X$, there is a prunable channel $i \in p(X)$ such that $x/i \in X/i$.
\end{lemma}

\begin{proof}
    Given a sequence $x \in X$, by the definition of well-behavedness, there is a threshold set $T(y)$ such that $x \in T(y) \subseteq X$.
    For a given threshold set $T(y)$ there is a unique $i$ such that $y_i = n$.
    We then have $T(y)/i = \{\, t/i \mid t \in T(y),\, t_i = 1 \,\}$.

    Here $t_i = 1$ holds for every contained sequence $t \in T(y)$ apart from the all zero sequence $t = (0, \ldots, 0)$.
    Since pruning the all zero sequence is the same as pruning the single nonzero element of a one-hot sequence, $(0, \ldots, 0)/i = \ohot in/i$, and that one-hot sequence is contained in the threshold set, $\ohot in \in T(y)$, we can drop the $t_i = 1$ condition and get $T(y)/i = \{\, t/i \mid t \in T(y) \,\}$.

    From this we get $x/i \in T(y)/i$ and since $T(y)/i \subseteq X/i$ also $x/i \in X/i$.
\end{proof}

\begin{corollary}[Well-Behaved Prunable Channels] \label{fact:wbprunable}
    Every nonempty well-behaved sequence set $X \in \cW$ has a nonempty set of prunable channels $p(X)$.\qed
\end{corollary}

For well-behaved sequence sets, we have a simple test for whether they can be sorted using only exchanges:

\begin{lemma}[Well-Behaved Sortedness Test] \label{fact:wbnonempty} For a well-behaved sequence set $X$, it has a minimal size $s(X) = 0$ for sorting if and only if $|X| \le n + 1$.
\end{lemma}

\begin{proof}
    One direction follows from the pigeon hole bound (\autoref{fact:nonempty}) which tells us that $|X| \le n + 1$ if $s(X) = 0$.

    For the converse direction, note that when $|X| \le n + 1$, $X$ is the union of at most one threshold set, as every threshold set contains $n + 1$ sequences.
    If $X$ is the empty set, the empty network vacuously sorts $X$. Otherwise $X$ is a threshold set which by definition is permutation-similar to the set of sorted sequences $(B^n)^\mathbf s$.
    In that case, as $s((B^n)^\mathbf s) = 0$, we get $s(X) = 0$ from invariance under permutation (\autoref{fact:permute}).
\end{proof}

We still need to show that well-behaved sequence sets are closed under the actions of comparators, exchanges and Boolean negation.
We also need to deal with ill-behaved sets we might get as the result of pruning a well-behaved set.

First we define an operation that produces the largest well-behaved sequence set contained in a given sequence set:

\begin{definition}[Well-Behaved Interior] \label{def:wbint}
    Given a sequence set $X$ we define its \emph{well-behaved interior}, denoted by $X^\circ$, as the largest well-behaved sequence set contained in $X$. As well-behaved sequence sets are by definition closed under union, this is well-defined and equivalent to the union of all threshold sets contained in $X$.
\end{definition}

As the well-behaved interior $X^\circ$ is contained in $X$, there is a variant of the Huffman Bound that uses the well-behaved interior of pruned sequence sets:

\begin{corollary}[The Well-Behaved Huffman Bound] \label{fact:huffmanwp}
    Given a well-behaved nonempty sequence set $X \subseteq \mathcal W \setminus \{\emptyset\}$, the minimal size for sorting $X$ is bounded by $s(X) \ge H_{1{+}\mathrm{max}}\{\, s((X / i)^\circ) \mid i \in p(X) \,\}_\#$.
\end{corollary}

\begin{proof}
    As $(X / i)^\circ \subseteq X/i$ for all $X$ and prunable $i$, this is a direct consequence of the Huffman bound (\autoref{fact:huffmanbound}), the monotonicity of Huffman's algorithm (\autoref{fact:huffmanmono}) and the monotonicity of the minimal size for sorting a sequence set (\autoref{fact:mono}).
    The condition $p(X) \ne \emptyset$ necessary for the Huffman bound follows from $X$ being well-behaved and nonempty (\autoref{fact:wbprunable}).
\end{proof}

To show that well-behaved sequence sets are closed under the other required operations, we first show that this is case for the threshold sets from which they are built.

\begin{lemma}[Threshold-Set Operations]
    Given a threshold set $X \in \mathcal T_n$ we have $X^{[i, j]} \in \mathcal T_n$, $X^{(i, j)}\in \mathcal T_n$ and $X^{\mathbf b} \in \mathcal T_n$ for all comparators $[i,j] \in C_n$ and all exchanges $(i, j) \in E_n$.
\end{lemma}

\begin{proof}
    The action of a comparator or an exchange $c$ on a threshold set $T(x)$ with $x=(1, \ldots, n)^\sigma$ results in $T(x)^c = T(x^c)$ which is again a threshold set. Applying Boolean negation to $T(x)$ results in $T(x)^\mathbf b = T(y)$ where $y=(n, \ldots, 1)^\sigma = (1, \ldots, n)^{\rho\sigma}$ and $\rho$ is the permutation that reverses the order of channels.
\end{proof}

Next we extend this to well-behaved sets.

\begin{lemma}[Well-Behaved Operations] \label{fact:wbclose}
    Given a well-behaved sequence set $X \in \mathcal W_n$ we have $X^{[i, j]} \in \mathcal W_n$, $X^{(i, j)} \in \mathcal W_n$ and $X^{\mathbf b} \in \mathcal W_n$ for all comparators $[i,j] \in C_n$ and all exchanges $(i, j) \in E_n$.
\end{lemma}

\begin{proof}
    Given a $c \in C_n \cup E_n \cup \{\mathbf b\}$, the action $X^c$ is, by definition, the elementwise action of $c$ on the sequences of $X$ and thus preserves unions.
    For a well-behaved sequence set $X$, this allows us to let $c$ act on each threshold set contained in $X$:
    \[
        X^c = \bigcup\,\{\, T(y)^c \mid y \in R_n,\, T(y) \subseteq X\,\}.
    \]

    We also saw, that the action of $c$ maps threshold sets to threshold sets.
    As the function $T$ defining threshold sets is a bijection from $R_n$ to $\mathcal T_n$, we can define an action of $c$ on the permutation sequences in $R_n$ such that for every sequence $y \in R_n$ we have $T(y)^c = T(y^c)$ and
    \[
        X^c = \bigcup\,\{\, T(y^c) \mid y \in R_n,\, T(y) \subseteq X\,\},
    \]
    which is a union of threshold sets and thus a well-behaved sequence set.

\end{proof}

As these operations include Boolean negation and exchanges, they include the operations whose orbits define similar sequence sets.

\begin{corollary}[Well-Behavedness of Similar Sequence Sets] \label{fact:wbpart}
    Given a well-behaved sequence set $X \in \cW_n$, we also have $Y \in \cW_n$ for all similar sequence sets $Y \sim X$. \qed
\end{corollary}

\section{An Algorithm for Computing $s(n)$} \label{sec:algo}

We obtain an algorithm for computing $s(n)$ from a more general algorithm that can compute $s(X)$ for every well-behaved sequence set $X \in \cW_n$, which includes $s(B^n) = s(n)$.

The algorithm for computing $s(X)$ for well-behaved $X$ uses dynamic programming.
We will start with a simple dynamic programming algorithm and then present several improvements until we reach the algorithm used to compute $s(11)$ and $s(12)$.

\subsection{Dynamic Programming}

For dynamic programming to be applicable, the problem under consideration needs to exhibit optimal substructure.
In our case this means that computing $s(X)$ for well-behaved sequence set $X$ must be reduced to a set of subproblems that again each consist of computing $s(Y)$ for another well-behaved sequence set $Y$.
Each of these problems must be smaller by a suitable measure.

Given a well-behaved sequence set $X \in \cW_n$ we can quickly test whether $s(X) = 0$ by checking for $|X| \le n + 1$ (\autoref{fact:wbnonempty}). This provides us with a base-case.

For $s(X) \ge 1$, using the successor recurrence (\autoref{fact:succ}) we can compute $s(X)$ from $s(X^{[i,j]})$ for each successor $X^{[i,j]} \in \mathcal U(X)$. As comparators maintain well-behavedness (\autoref{fact:wbclose}) we know that every successor of $X$ is also well-behaved.

We have also seen that the well-behaved sequence sets are partitioned into classes of similar sequence sets (\autoref{fact:wbpart}), where for every two similar sequence sets $X \sim Y$ we have $s(X) = s(Y)$ (\autoref{fact:sim}).
Thus it would be wasteful to consider both $s(X)$ and $s(Y)$ as distinct subproblems.
We can avoid that by defining a canonical representative for each class of similar sequence sets.

\begin{definition}
	Given a sequence set $X$ the \emph{canonical representative similar to $X$} is written as $X^\sim$ and defined as an arbitrary but fixed representative of the equivalence class $[X]_\sim$.
	Given a family of sequence sets $\mathcal F$, the family of canonical representatives $\{\, X^\sim \mid X \in \mathcal F\,\}$ is denoted by $\mathcal F^\sim$.
\end{definition}

Now, whenever we are asked to compute $s(X)$, we instead compute $s(X^\sim)$, increasing the overlap of subproblems.
We will give more details on how to compute the canonical representative $X^\sim$ similar to a given $X$ in \autoref{sec:implement}.

Applying memoization to the successor recurrence while using canonicalization gives us the following simple but not very efficient dynamic programming algorithm:

\begin{algo}[Memoizing Dynamic Programming Algorithm] \label{alg:memodp}
	\algorithm{
		global $t \gets (X \mapsto$ if $s(X) = 0$ then $0$ else $\top)$

		\cmnt{$\uparrow$ Here $s(X) = 0$ if and only if $|X| \le n + 1$ for $X \in \cW_n^\sim$}
		\AW
		procedure $\proc{MemoMinSize}(X)\colon$ \cmnt{$X \in \cW^\sim$} \AC{
			if $t(X) = \top\colon$ \cmnt{This can only happen for $s(X) > 0$} \AC{
				$t(X) \gets 1 + \min\,\{\, \proc{MemoMinSize}(Y) \mid Y \in \mathcal U(X)^\sim \,\}$
			}

			return $t(X)$
		}
	}
\end{algo}

Note that we do not perform an explicit check for $s(X) = 0$ as part of \proc{MemoMinSize} but instead initialize $t(X) = 0$ for all $s(X) = 0$.
This allows us to use a single check for both $s(X) = 0$ and for whether we already computed a $s(X) \ne 1$.
An implementation would only store those $t(X)$ that have values that differ from the initial value, as unchanged values can be computed on the fly whenever necessary.
This makes such an initialization practical.

Doing this might seem unnecessarily convoluted for a simple algorithm like this, but it will simplify our extensions of this algorithm.

To prove correctness of this algorithm, we need to show that \proc{MemoMinSize} always terminates.
For this, we introduce a suitable measure on \proc{MemoMinSize}'s argument which strictly decreases during the recursion:

\begin{definition}[Sequence Set Weights]
	Let $w(X)$ denote the \emph{weight} of a sequence set $X \subseteq B^n$, defined to be $\sum_{x, y \in X} \Delta(x, y)$ where $\Delta(x, y)$ is the Hamming distance between $x$ and $y$ given by $\Delta(x, y) = |\{\, i \in [n] \mid x_i \ne y_i\,\}|$.
\end{definition}

\begin{lemma}[Successor Weights] \label{fact:succweight}
	Given a sequence set $X \subseteq B^n$, for every successor $X^{[i,j]} \in \mathcal U(X)$, we have a strictly decreasing weight $w(X^{[i,j]}) < w(X)$.
\end{lemma}

\begin{proof}
	Given two sequences $x$ and $y$ the inequality $\Delta(x^{[i, j]}, y^{[i, j]}) \le \Delta(x, y)$ holds for all $[i, j] \in C_n$. This can be verified by checking the 16 possible combinations of values for $x_i, x_j, y_i, y_j$.

	Now $w(X^{[i, j]}) \le w(X)$ holds for all $[i, j] \in C_n$ as \[\sum_{x, y \in X^{[i, j]}} \Delta(x, y) \le \sum_{x, y \in X} \Delta(x^{[i, j]}, y^{[i, j]}) \le \sum_{x, y \in X} \Delta(x, y).\]

	We can also see that for $w(X^{[i, j]}) = w(X)$ to hold, we need $\Delta(x^{[i, j]}, y^{[i, j]}) = \Delta(x, y)$ for all $x, y \in X$.
	This is exactly the case when the comparator $[i, j]$ swaps the values of $i$ and $j$ for either all or for none of the $x \in X$ which happens exactly if $X^{[i, j]} = X$ or $X^{[i, j]} = X^{(i, j)}$.
	In both cases the sequence sets are permutation-similar, $X^{[i, j]} \permsim X$, so for a successor $X^{[i, j]} \in \mathcal U(X)$ we have $w(X^{[i,j]}) < D(X)$.
\end{proof}

As we perform canonicalization, we also need to make sure that it does not change the weight of a sequence set:

\begin{lemma}[Weights of Similar Sequence Sets]
	Given two sequence sets $X, Y \subseteq B^n$ with $X \sim Y$, we have $w(X) = w(Y)$
\end{lemma}

\begin{proof}
	Follows from $\Delta(x, y) = \Delta(x^{(i, j)}, y^{(i, j)})$ for every $(i, j) \in E_n$ and $\Delta(x, y) = \Delta(x^\mathbf b, y^\mathbf b)$
\end{proof}

With this we can show that our algorithm is correct:

\begin{lemma}[Correctness of \proc{MemoMinSize}]
	For all $X \in \cW^\sim$, $\proc{MemoMinSize}(X)$ terminates and returns $s(X)$.
\end{lemma}

\begin{proof}
	For the initial value of the map $t$, we have either $t(X) = s(X)$ or $t(X) = \top$ for every $X \in \cW^\sim$ and we have $t(X) = 0$ if $s(X) = 0$.
	We will show that this is an invariant of the algorithm.

	The non-negative integer weight of \proc{MemoMinSize}'s argument strictly decreases with every recursive invocation, thus the algorithm terminates.
	By induction we can assume that this invariant is maintained when recursively invoking \proc{MemoMinSize}.

	For an $X$ with $t(X) = \top$, by the invariant, $s(X) > 0$, which is the condition of the successor recurrence.
	When \proc{MemoMinSize} assigns to $t(X)$ it assigns the value $s(X)$ computed from the successor recurrence (\autoref{fact:succ}) with additional canonicalization, assuming the recursive calls $\proc{MemoMinSize}(Y)$ return $s(Y)$.

	The recursive calls receive a well-behaved argument as they are constructed from $X$ using operations closed under well-behavedness (\autoref{fact:wbclose}).
	Therefore correct results of the recursive calls follow from induction.

	Due to invariance under similarity (\autoref{fact:sim}), the added canonicalization in the successor recurrence does not change the result.

	Therefore all assignment made to $t(X)$ maintain the invariant.
	Finally, when \proc{MemoMinSize} returns, we have $t(X) \ne \top$ thus get a return value of $t(X) = s(X)$ from our invariant.
\end{proof}

\subsection{Successive Approximation}

By moving from memoization based dynamic programming to successive approximation based dynamic programming, we obtain an improved algorithm that we can later extend to make use of the Huffman bound.

To do so we replace the map $t$ in \autoref{alg:memodp}, which stores the values of $s$ already computed, with the map $b$ which stores bounding intervals such that it is a pointwise over-approximation of $s$ on canonical well-behaved sets, i.e.\ such that $s(X) \in b(X)$ for all $X \in \mathcal W^\sim$.
Below we will write $s \pwin b$ and $b' \pwsubseteq b$ for such pointwise membership and inclusion, respectively, both restricted to the domain of the right hand side $b$.

The algorithm progresses by repeatedly selecting a sequence set $X$ and updating the interval $b(X)$ by moving its bounds closer to $s(X)$, improving the approximation.
Repeatedly updating $b(X)$ eventually turns it into the singleton interval $\{s(X)\}$, in which case we say that $X$ is \emph{fathomed}.

Similar to how we initialized $t$, we will use the well-behaved sortedness test (\autoref{fact:wbnonempty}) to initialize $b(X) = \{0\}$ for all $X \in \cW^\sim$ with $s(X) = 0$.
This ensures that for any non-fathomed $X$ we have $s(X) \ge 1$.

To avoid infinite intervals we will use an arbitrary upper bound $\hat s(n) \ge s(n)$ and initialize $b(X) = \{1, \ldots, \hat s(n)\}$ for every $X \in \cW^\sim$ with $s(X) \ge 1$. This initial state of the map $b$ is
denoted $b_0$.
As before, an implementation would only store values of $b$ for keys that differ from the corresponding values in $b_0$.

For $\hat s(n)$ it is advantageous to use the best known upper bounds for $s(n)$ as this means the algorithm does not have to rediscover these.

In the memoization variant we always recurse on all subproblems before computing and storing the value for the current problem.
In a successive approximation approach, the updates to approximation $b$ and the scheduling of subproblems are decoupled.
This freedom makes it possible to heuristically schedule exploration of subproblems in a way that attempts to make the most progress with the least amount of work required.

We first define a procedure which updates $b(X)$ for a given $X \in \cW^\sim$ based on other current values of $b$.
The update replaces $b(X)$ with the best approximation of $s(X)$ derivable from the successor recurrence under the assumption that $s \pwin b$.

To get this best approximation of $s(X)$, we compute the (canonicalized) successor recurrence (\autoref{fact:succ}) with all right hand side uses of $s$ replaced with an arbitrary function $a$ that is pointwise contained in our current approximation $b$.
We then take the set containing all possible values of this modified successor recurrence for different values of $a$.
This is the set $B = \{\, 1 + \min\, \{\, a(Y) \mid Y \in \mathcal U(X)^\sim \,\} \mid a \pwin b \,\}$.
It is a valid over-approximation of $s(X)$ as, by assumption, $s$ itself is among the possible choices for $a$.

We then intersect $b(X)$ with the smallest interval containing $B$, i.e. $b(X) \gets b(X) \cap \{ \min B, \ldots, \max B \}$.
As $1 + \min\, \{\, a(Y) \mid Y \in \mathcal U(X)^\sim \,\}$ is monotone in each $a(Y)$, we have $\min B = 1 + \min\, \{\, \min b(Y) \mid Y \in \mathcal U(X)^\sim \,\}$ and $\max B = 1 + \min\, \{\, \max b(Y) \mid Y \in \mathcal U(X)^\sim \,\}$.

\begin{algo}
	\algorithm{
		{\normalfont\emph{arbitrarily choose $\hat s(n) \ge s(n)$ for all $n$}}

		global $b \gets b_0 = (X \mapsto$ if $s(X) = 0$ then $\{0\}$ else $\{1, \ldots, \hat s(n)\})$

		\cmnt{$\uparrow$ Here $s(X) = 0$ if and only if $|X| \le n + 1$ for $X \in \cW_n^\sim$}

		\AW

		procedure $\proc{SuccessorUpdate}(X)\colon$ \AC{
			if $X$ is not fathomed$\colon$ \cmnt{Also imples $s(X) \ne 0$ due to the initialization of $b$} \AC{
				for $f$ in $\{\min, \max\}\colon$ $B_f \gets 1 + \min\, \{\, f(b(Y)) \mid Y \in \mathcal U(X)^\sim \,\}$

				$b(X) \gets b(X) \cap \{B_{\min}, \ldots, B_{\max}\}$ \cmnt{As $B_{\min} \le s(X) \le B_{\max}$}
			}
		}
	}
\end{algo}

As we perform the update by intersecting with the best approximation, by construction, we get the following invariant and progress guarantee:

\begin{lemma}[\proc{SuccessorUpdate} Invariant] \label{fact:succupdateinv}
	Given a global state $b = b_u$ which is a point-wise over-approximation of $s$, i.e.\ $s \pwin b_u \pwsubseteq b_0$, performing a $\proc{SuccessorUpdate}(X)$ for an $X \in \cW^\sim$ results in a global state $b = b_v$ which is also a point-wise over-approximation of $s$ and which is a refinement of the previous global state, i.e.\ $s \pwin b_v \pwsubseteq b_u \pwsubseteq b_0$. \qed
\end{lemma}

\begin{lemma}[\proc{SuccessorUpdate} Progress] \label{fact:succupdateprog}
	Given a global state $b = b_u$ with $s \pwin b_u \pwsubseteq b_0$ where all canonical successors $U(X)^\sim$ of $X$ are fathomed, performing a $\proc{SuccessorUpdate}(X)$ for an $X \in \cW^\sim$ results in a global state $b = b_v$ where $X$ is fathomed. \qed
\end{lemma}

The algorithm below uses this update recursively to improve any non-fathomed $X$.
It allows for an arbitrary heuristic choice in two places.
The first choice is whether to further improve the bound after an initial improvement when $X$ is not yet fathomed.
The second choice is which subset of successors to recurse on between two \proc{SuccessorUpdate} calls.

\begin{algo}
	\algorithm{
		procedure $\proc{SuccessorImprove}(X)\colon$ \AC{
			if $X$ is not fathomed$\colon$ \cmnt{Otherwise there is nothing to improve} \AC{
				$c \gets b(X)$

				loop$\colon$ \AC{
					$\proc{SuccessorUpdate}(X)$

					if $b(X) \ne c$ and $(X$ is fathomed or {\normalfont\emph{arbitrary choice}}$)\colon$ return

					$\mathcal A \gets \{\, Y \in \mathcal U(X)^\sim \mid Y \textbf{ is not fathomed} \,\}$

					{\normalfont\emph{arbitrarily choose a non-empty $\mathcal B \subseteq \mathcal A$}}

					for $Y \in \mathcal B\colon$ $\proc{SuccessorImprove}(Y)$
				}
			}
		}
	}
\end{algo}

This algorithm maintains the same invariant as \proc{SuccessorUpdate}.
It also guarantees progress when invoked for an arbitrary non-fathomed $X$.

\begin{lemma}[Correctness of \proc{SuccessorImprove}] \label{fact:succimprovecorrect}
	Given a global state $b = b_u$ with $s \pwin b_u \pwsubseteq b_0$, performing $\proc{SuccessorImprove}(X)$ for a non-fathomed $X \in \cW^\sim$ results in a global state $b = b_v$ with a strictly improved bound $b_v(X) \subset b_u(X)$ and which is a refinement of the previous global state, $s \pwin b_v \pwsubseteq b_u \pwsubseteq b_0$.
\end{lemma}

\begin{proof}
	First we observe that, as before in \proc{MemoMinSize}, the recursion depth is limited due to the strictly decreasing sequence set weight (\autoref{fact:succweight}). Thus by induction we can assume correctness holds for all recursive \proc{SuccessorImprove} calls.

	If the procedure terminates, we get $s \pwin b_v \pwsubseteq b_u \pwsubseteq b_0$ from the correctness assumption of recursive calls and the \proc{SuccessorUpdate} invariant.

	We only return when we have $b_v(X) \ne b_u(X)$, so from $b_v \pwsubseteq b_u$ we get $b_v(X) \subset b_u(X)$.

	The choice for $\mathcal B$ is always possible, as $\mathcal A = \emptyset$ could only happen when the progress condition of \proc{SuccessorUpdate} holds (\autoref{fact:succupdateprog}).
	In that case \proc{SuccessorUpdate} would have fathomed $X$ by setting $b(X) \gets \{s(X)\}$ which would have caused us to immediately return before reaching this choice.

	Finally, to show termination, we can observe that at least one $b(Y)$ for a $Y \in \mathcal A$ gets smaller every iteration.
	This follows from the correctness assumption of the recursive calls and the nonempty choice of $\mathcal B$.
	As all $b(Y)$ for $Y \in \mathcal U(X)^\sim$ are finite and can only decrease in size, eventually all canonical successors will be fathomed.
	This is again the progress condition (\autoref{fact:succupdateprog}) that causes \proc{SuccessorUpdate} to fathom $X$ which then causes an immediate return from \proc{SuccessorImprove}.
\end{proof}

\begin{corollary}
	Given a well-behaved sequence set $X \in \cW$ we can compute $s(X)$ using \inlinealg{\;$b \gets b_0;$}\inlinealg{while}\inlinealg{$|b(X^\sim)| > 1\colon$}\inlinealg{$\,\{\proc{SuccessorImprove}(X^\sim)\}$\;}\ resulting in $b(X^\sim) = \{s(X)\}$.
\end{corollary}

\begin{proof}
	Every call to $\proc{SuccessorImprove}(X^\sim)$ will decrease $|b(X^\sim)|$ until $X^\sim$ is fathomed and we have $b(X^\sim) = \{s(X^\sim)\} = \{s(X)\}$.
\end{proof}

With a suitable heuristic this already improves upon the memoizing algorithm.
When a sequence sets $X$ has two successors $Y, Z \in \mathcal U(X)$ with $\max b(Y) \le \min b(Z)$, the successive approximation algorithm never has to improve the bounds for $Z$ as $s(Y) \le s(Z)$.

\subsection{Using the Huffman Bound}

Apart from freedom in choosing \proc{SuccessorImprove}'s heuristics, the successive approximation algorithm offers flexibility by allowing extension of the bound improvement strategy with any kind of update that maintains the same invariant as \proc{SuccessorUpdate} does (\autoref{fact:succupdateinv}).
This is how we will incorporate the Huffman bound in our algorithm.

First we are going add canonicalization to the well-behaved Huffman bound (\autoref{fact:huffmanwp}), denoting it with $h$, obtaining $s(X) \ge h = H_{1{+}\mathrm{max}}\,\{\, s((X/i)^{\circ\sim}) \mid i \in p(X) \,\}_\#$ given a non-empty well-behaved $X \subseteq \cW^\sim \setminus \{\emptyset\}$.
Following the same approach as for the successor recurrence, we are going to construct a set $H$ from the canonicalized well-behaved Huffman bound.
We replace every right hand side use of $s$ with an arbitrary $a$ pointwise contained in our current approximation $b$ of $s$.
From this we obtain a set of possible values that contains $h$: \[h \in H = \{\, H_{1{+}\mathrm{max}}\,\{\,a((X / i)^{\circ\sim}) \mid i \in p(X) \,\}_\# \mid a \pwin b \,\}.\]
Since Huffman's algorithm is monotone (\autoref{fact:huffmanmono}), we can again efficiently compute the bounds of the smallest interval containing $H$:
\begin{align*}
	\min H & = H_{1{+}\mathrm{max}}\,\{\, \min b((X/i)^{\circ\sim}) \mid i \in p(X) \,\}_\# \quad\text{ and } \\
	\max H & = H_{1{+}\mathrm{max}}\,\{\, \max b((X/i)^{\circ\sim}) \mid i \in p(X) \,\}_\#.
\end{align*}

The only update that we can safely perform is intersecting $b(X)$ with $\{\min H, \ldots\}$ as $h \in H$ is only a lower bound for $s(X)$.
Nevertheless computing $\max H$ is useful to us. If $\max H \le \min b(X)$ holds, we discover that performing more work to decrease the size of $H$ is unnecessary, as our bounds $b(X)$ are already as good as we can hope to obtain from the Huffman bound for this $X$.
We will use the return value to notify the caller when this happens.

In addition to applying the Huffman bound to $X \in \cW^\sim$, we can also apply it to the negated sequence set $X^\mathbf b$ which in general might get us a different bound.

Note that we don't have to test for $X = \emptyset$ when using the well-behaved Huffman bound, as $s(\emptyset) = 0$ and thus $\emptyset$ is fathomed.

This results in the following update algorithm:

\begin{algo} \label{alg:huffmanupdate}
	\algorithm{
		procedure $\proc{HuffmanUpdate}(X)\colon$ \AC{
			$u \gets -\infty$

			if $X$ is not fathomed$\colon$ \cmnt{Otherwise there is nothing to update} \AC{
				for $Z$ in $\{X, X^\mathbf b\}\colon$ \AC{
					for $f$ in $\{\min, \max\}\colon$ $H_f \gets H_{1{+}\mathrm{max}}\,\{\, f(b((Z/i)^{\circ\sim})) \mid i \in p(X) \,\}_\#$

					$b(X) \gets b(X) \cap \{H_{\min}, \ldots\}$ \cmnt{As $s(X) \ge H_{\min}$}

					$u \gets \max\,\{u, H_{\max}\}$
				}
			}

			return $\min b(X) \ge u$
		}
	}
\end{algo}

Again, as we constructed \proc{HuffmanUpdate} such that it compute a sound approximation, we get the following invariant:

\begin{lemma}[\proc{HuffmanUpdate} Invariant] \label{fact:huffmanupdateinv}
	Given a global state $b = b_u$ which is a point-wise over-approximation of $s$, i.e.\ $s \pwin b_u \pwsubseteq b_0$, performing a $\proc{HuffmanUpdate}(X)$ for an $X \in \cW^\sim$ results in a global state $b = b_v$ which is also a point-wise over-approximation of $s$ and which is a refinement of the previous global state, i.e.\ $s \pwin b_v \pwsubseteq b_u \pwsubseteq b_0$. \qed
\end{lemma}

To enable the algorithm using \proc{HuffmanUpdate} to make progress we ensure that it will return \textbf{true} when all involved sequence sets are fathomed.

\begin{lemma}[\proc{HuffmanUpdate} Progress] \label{fact:huffmanupdateprog}
	Given a global state $b$ with $s \pwin b \pwsubseteq b_0$ and where $(Y/i)^{\circ\sim}$ is fathomed for all $i \in p(Y)$ and all $Y \in \{X, X^\mathbf b\}$, performing a $\proc{HuffmanUpdate}(X)$ will result in a return value of {\normalfont\textbf{true}}.
\end{lemma}

\begin{proof}
	In this case the Huffman bound for $X$ and for $X^\mathbf b$ will be computed exactly, so for both $H_{\min} = H_{\max}$ and thus after the update $\min b(X) \ge u$.
\end{proof}

Before presenting a new improvement routine that combines \proc{SuccessorUpdate} and \proc{HuffmanUpdate}, we will introduce a third way to update bounds.

\subsection{Pruning Channels}

If we are given a sequence set $X \in \cW$ that has only a single prunable channel $p(X) = \{i\}$, we can do better than either the \proc{SuccessorUpdate} or the \proc{HuffmanUpdate}.

We start with an observation about threshold sets:

\begin{lemma}
	Given a threshold set $T(y)$ with $y \in R_n$ where $y_i = n$, there is only a single prunable channel $p(T(y)) = \{i\}$ and $T(y)/i = T(y/i)$.
\end{lemma}

\begin{proof}
	For a threshold set $T(y)$ with $y_i = n$ we have $\ohot jn \in T(y)$ exactly if $j = i$, so $i$ is the only prunable channel of $T(y)$.

	When pruning that channel, we can substitute the definition of threshold sets (\autoref{def:thresholdsets}) and exchange the order of pointwise threshold function application and pruning of sequences to get
	\begin{align*}
		T(y)/i &= \{\,y^{t_k}/i \mid 1 \le k \le n + 1, (y^{t_k})_i = 1\,\} \\
		&= \{\,(y/i)^{t_k} \mid 1 \le k \le n + 1,\, y_i \ge k\,\} \\
		&= \{\,(y/i)^{t_k} \mid 1 \le k \le n\,\} = T(y/i).
	\end{align*}
\end{proof}

With this, we get the following property for all well-behaved sequence set with only a single prunable channel:

\begin{lemma}[Unique Prunable Channel]
	Given a well-behaved sequence set $X \in \cW_n$ with $|p(X)| = 1$, the minimal size for sorting $X$ does not change when pruning $X$, i.e. $s(X) = s((X/i)^\circ)$.
\end{lemma}

\begin{proof}
	First, we will show that when $X$ has only a single prunable channel $i \in p(X)$ and is well-behaved, the pruned sequence set $X/i$ is also well-behaved and thus $(X/i)^\circ = X/i$.

	As the set of prunable channels preserves unions, and every threshold set has exactly one prunable channel, we know that every threshold set $T(y) \subseteq X$ has the same unique prunable channel.

	The pruning operation on sequence sets also preserves unions, so we can perform it individually on each contained threshold set.
	Using the previous lemma, and the fact that all contained threshold sets have the prunable channel $i$, we then get \[
		X/i = \bigcup\, \{\, T(y)/i \mid y \in R_n,\, T(y) \subseteq X\,\} = \bigcup\, \{\, T(y/i) \mid y \in R_n,\, T(y) \subseteq X\,\},
	\]
	which as a union of threshold sets is a well-behaved sequence set, so $(X/i)^\circ = X/i$.

	Next we will show that $s(X) = s(X/i)$.
	As $s(X) \ge H_{1{+}\mathrm{max}}\,\{s(X/i)\}_\# = s(X/i)$ follows from the Huffman bound (\autoref{fact:huffmanbound}) we only need to show $s(X) \le s(X/i)$.

	Given a partial sorting network on $c$ on $X/i$ we can construct a partial sorting network $d$ on $X$ having the same size by inserting an input channel such that it becomes channel $i$ and directly connecting it to a new last output channel.
	As $x_i = 1$ for all $x \in X$ with $x \ne (0, \ldots, 0)$ the output of $d$ acting on an $x \in X$ is always sorted.
	Thus this construction gives us a partial sorting sorting network on $X$ having size $s(X/i)$ and thus $s(X) \le s(X/i)$.

	From $s(X) = s(X/i)$ and $(X/i)^\circ = X/i$ we then get $s(X) = s((X/i)^\circ)$.
\end{proof}

Therefore, with only one prunable channel $i \in p(X)$, we can always set $b(X)$ to the value of $b((X/i)^\sim)$.
This allows us to always improve $b(X)$ by first improving $b((X/i)^\sim)$ and then performing this update.
If $X$ has multiple prunable channels, it is still possible that $X^\mathbf b$ has only a single prunable channel $i$, in which case we instead set $b(X)$ to $b((X^\mathbf b/i)^\sim)$.

\subsection{Combining Strategies for Improving Bounds}

With this we can now present a bound improvement routine combining these three updates, giving us the final version of our algorithm, presented as the following pseudocode:

\begin{algo}[Successive Approximation Algorithm for $s(X)$]
	\algorithm{
		global $b \gets b_0 = \cdots$		\cmnt{As above}

		procedure $\proc{SuccessorUpdate}(X)\colon \cdots$

		procedure $\proc{HuffmanUpdate}(X)\colon  \cdots$
		}\vskip-1.75em\algorithm{
		procedure $\proc{Improve}(X)\colon$ \AC{
			if $X$ is not fathomed$\colon$ \cmnt{Otherwise there is nothing to improve} \AC{
				$c \gets b(X)$

				if $\proc{PrunedImproveStep}(X, c)\colon$ return

				if $\proc{HuffmanImproveStep}(X, c)\colon$ return

				$\proc{SuccessorImproveStep}(X, c)$
			}
		}
		}\vskip-1.75em\algorithm{
		procedure $\proc{PrunedImproveStep}(X, c)\colon$ \cmnt{Returns whether it made progress} \AC{
			for $Z$ in $X^\ang{\mathbf b}$ if $|p(Z)| = 1\colon$ \AC{
				$Y \gets (Z/i)^\sim$ where $i \in p(Z)$

				loop$\colon$ \AC{
					$b(X) \gets b(X) \cap b(Y)$ \cmnt{Update using the fact that $s(X) = s(Y)$}

					\cmnt{Each iteration improves $Y$, guaranteeing we will eventually return}

					if $b(X) \ne c$ and $(X$ is fathomed or {\normalfont\emph{arbitrary choice}}$)\colon$ return true

					$\proc{Improve}(Y)$
				}
			}

			return false
		}
		}\vskip-1.75em\algorithm{
		procedure $\proc{HuffmanImproveStep}(X, c)\colon$ \cmnt{Returns whether it made progress} \AC{
			loop$\colon$ \AC{
				$\mathit{done} \gets \proc{HuffmanUpdate}(X)$

				if $b(X) \ne c$ and $(X$ is fathomed or {\normalfont\emph{arbitrary choice}}$)\colon$ return true

				if $\mathit{done}\colon$ return false

				$\mathcal A \gets \{\, (Y/i)^{\circ\sim} \mid Y \in \{X, X^{\mathbf b\sim}\},\, i \in p(Y),\, (Y/i)^{\circ\sim} \textbf{ is not fathomed} \,\}$

				{\normalfont\emph{arbitrarily choose $\mathcal B \subseteq \mathcal A$ such that $|\mathcal B| \ge 1$}}

				for $Y \in \mathcal B\colon$ $\proc{Improve}(Y)$
			}
		}
		}\vskip-1.75em\algorithm{
		procedure $\proc{SuccessorImproveStep}(X, c)\colon$ \cmnt{Guaranteed to make progress} \AC{
			loop$\colon$ \AC{
				$\proc{SuccessorUpdate}(X)$

				if $b(X) \ne c$ and $(X$ is fathomed or {\normalfont\emph{arbitrary choice}}$)\colon$ return

				$\mathcal A \gets \{\, Y \in \mathcal U(X)^\sim \mid Y \textbf{ is not fathomed}\,\}$

				{\normalfont\emph{arbitrarily choose $\mathcal B \subseteq \mathcal A$ such that $|\mathcal B| \ge 1$}}

				for $Y \in \mathcal B\colon$ $\proc{Improve}(Y)$
			}
		}
		}\vskip-1.75em\algorithm{
		procedure $\proc{MinSize}(X)\colon$ \AC{
			while $X^\sim$ is not fathomed$\colon$ $\proc{Improve}(X^\sim)$

			return $k \in b(X^\sim)$
		}
	}
\end{algo}

\begin{theorem}[Correctness of \proc{Improve}]
	Given a global state $b = b_u$ with $s \pwin b_u \pwsubseteq b_0$, performing $\proc{Improve}(X)$ for a non-fathomed $X \in \cW^\sim$ results in a global state $b = b_v$ with a strictly improved bound $b_v(X) \subset b_u(X)$ and which is a refinement of the previous global state, $s \pwin b_v \pwsubseteq b_u \pwsubseteq b_0$.
\end{theorem}

\begin{proof}
	Here \proc{Improve} is mutually recursive with \proc{Pruned\-Improve\-Step}, \proc{Huffman\-Improve\-Step} and \proc{Successor\-Improve\-Step}.
	To justify induction, we can observe that all recurring $\proc{Improve}$ calls via \proc{PrunedImproveStep} or \proc{HuffmanImproveStep} reduce the length of the sequence set by one.
	A call via \proc{SuccessorImproveStep} keeps the length constant but strictly reduces the weight (\autoref{fact:succweight}).
	Thus infinite recursion is prevented by the lexicographically strictly decreasing pair of length and weight.

	All updates performed to a $b(X)$ by each of \proc{PrunedImproveStep}, \proc{HuffmanImproveStep} and \proc{SuccessorImproveStep} maintain $s(X) \in b(X)$ and only decrease the size of $b(X)$.
	Thus $s \pwin b_v \pwsubseteq b_u \pwsubseteq b_0$ holds.

	When \proc{PrunedImproveStep} or \proc{HuffmanImproveStep} return \textbf{true}, the size of $b(X)$ decreased and \proc{Improve} returns immediately.
	In those cases we also have $b_v(X) \subset b_u(X)$.

	When both \proc{PrunedImproveStep} and \proc{HuffmanImproveStep} return \textbf{false}, we perform a \proc{SuccessorImproveStep} which will terminate and ensure $b_v(X) \subset b_u(X)$ by the same reasoning we used to show correctness of \proc{SuccessorImprove} (\autoref{fact:succimprovecorrect}).

	For \proc{HuffmanImproveStep} we know that $\proc{HuffmanUpdate}(X)$ will return $\textbf{false}$ if $\mathcal A$ would be empty (\autoref{fact:huffmanupdateprog}).
	Following the same argument we used for \proc{SuccessorImprove}, we see that all candidates for $\mathcal A$ will eventually be fathomed, which forces $\proc{HuffmanUpdate}$ to return.

	Finally, when \proc{PrunedImproveStep} does not immediately return \textbf{false}, by induction we know that $\proc{Improve}(Y)$ will eventually fathom $Y$, causing \proc{PrunedImproveStep} to return \textbf{true}.
\end{proof}

As we saw before, repeatedly calling a procedure with this property, allows us to compute $s(X)$ for every well-behaved sequence set $X$.

\begin{corollary}
	Given a well-behaved sequence set $X \in \cW$, we can compute $\proc{MinSize}(X) = s(X)$. \qed
\end{corollary}

\section{Implementation} \label{sec:implement}

The algorithm described in the previous section performs several non-trivial operations on sequence sets, for which an efficient implementation is essential for the algorithm's overall performance.
In this section we will describe the algorithms used in the subroutines implementing these operations.

Even with efficient subroutines, computing $s(11)$ and $s(12)$ requires significant computing resources.
Thus we want to use an optimized implementation which makes efficient use of the available memory and of all threads of a multi-core CPU.
We give an outline of how to implement a parallel variant of the algorithm in \autoref{sec:parallel}, where we also mention further heuristic choices used in our implementation.
For more details, the implementation's source code is available at \cite{harderJixSortnetopt2020}.

\subsection{Canonicalization} \label{sec:canon}

Our algorithm makes heavy use of canonicalization, so it is essential that we have a fast procedure to compute the canonical sequence set similar to an arbitrary given sequence set.

As at most two sequence sets are similar under negation, it is sufficient to have a canonicalization procedure for similarity under permutation, which then can be performed on an $X$ and an $X^\mathbf b$.
We then select the smaller of both results by an arbitrary easily computed order on sequence sets.

To find a canonical representative for a class of sequence sets similar under permutation, we can use canonical graph labeling.
We can represent a sequence set $X \subseteq B^n$ as a hypergraph (allowing for empty and singleton edges) by defining a map $f$ that takes a given sequence to the set of channels that are set to one, i.e. $f(x) = \{\, i \in [n] \mid x_i = 1\,\}$.
The hypergraph for $X$ is the hypergraph on the vertices $[n]$ with hyperedges $f(X)$.

Now two given sequence sets $X, Y$ are permutation-similar, $X \permsim Y$, if and only if their hypergraphs are isomorphic. Thus if we are able to perform a canonical re-labeling of the hypergraph vertices, we get a permutation that takes every sequence set to a unique representative with respect to permutations.

A canonical re-labeling of a hypergraph can be performed by canonically re-labeling the bipartite incidence graph of the hypergraph and ignoring the labels of the vertices on the side used to represent the hyperedges.
This can be done with an existing open source graph canonicalization library like \emph{nauty} or \emph{Traces} by \textcite{mckayPracticalGraphIsomorphism2014}.

In our implementation, we use a canonical labeling algorithm that works directly on our sequence set representation.
It is based on the breadth-first search approach pioneered by Traces \cite{mckayPracticalGraphIsomorphism2014,pipernoSearchSpaceContraction2011}.
To simplify the implementation, we do not use orbits of the automorphism group for pruning, and instead only prune the subtree which led us to discover an automorphism.
We also use heuristics that are tuned for our use-case, e.g. we only individualize vertices of the incidence graph corresponding to the hypergraph's vertices, as these always form a base of the incidence graph's automorphism group.

\subsection{Well-Behaved Interior Computation}

We defined the well-behaved interior $X^\circ$ of a sequence set $X$ as the union of all contained threshold sets (\autoref{def:wbint}).

This definition is not suitable for a direct implementation as it would have to test all $n!$ different length-$n$ threshold sets.
Even iterating over only those $y \in R_n$ for which the corresponding threshold set $T(y)$ is contained in $X$ is something we would like to avoid.

Instead we are going to derive an efficient implementation of $X^\circ$ that works directly on the Boolean sequences contained in $X$ without decomposing them into a union of threshold sets.

Consider $B^n$ as a poset using the product order, i.e. for given sequences $x, y \in B^n$, we have $x \le y$ if and only if $x_i \le y_i$ for all $i$.
It is bounded below by $\bot = (0, \ldots, 0)$ and above by $\top = (1, \ldots, 1)$.

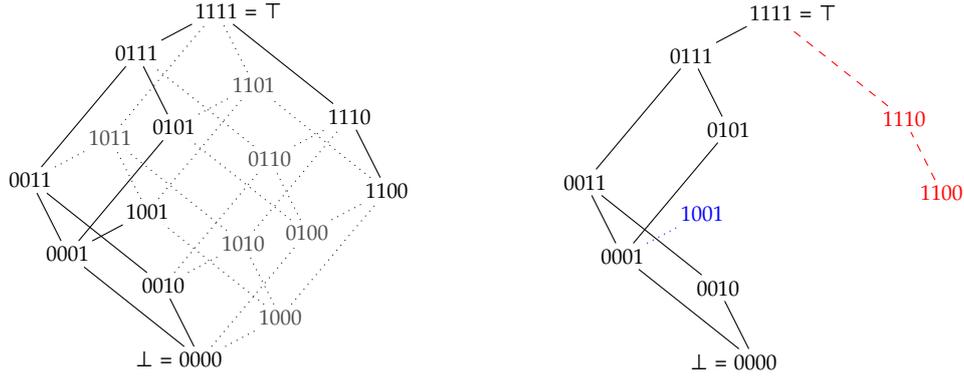
\begin{figure}
	\centering
\newcommand{\tessaract}{
    \foreach \x in {0, 1} {
        \foreach \y in {0, 1} {
            \foreach \z in {0, 1} {
                \foreach \w in {0, 1} {
                    \path ($
                        (1.5 * \x, 0.8 * \x) +
                        (2 * \y, 2.4 * \y) +
                        (-0.7 * \z, 1.4 * \z) +
                        (-2.5 * \w, 2 * \w) +
                        (0, 0)
                    $)
                        node[n\x\y\z\w] (n\x\y\z\w) {\scriptsize\x\y\z\w};
                }
            }
        }
    }

    \foreach \x in {0, 1} {
        \foreach \y in {0, 1} {
            \foreach \z in {0, 1} {
                \path[e\x\y\z x] (n\x\y\z0) edge (n\x\y\z1);
            }
        }
    }

    \foreach \x in {0, 1} {
        \foreach \y in {0, 1} {
            \foreach \w in {0, 1} {
                \path[e\x\y x\w] (n\x\y0\w) edge (n\x\y1\w);
            }
        }
    }

    \foreach \x in {0, 1} {
        \foreach \z in {0, 1} {
            \foreach \w in {0, 1} {
                \path[e\x x\z\w] (n\x0\z\w) edge (n\x1\z\w);
            }
        }
    }

    \foreach \y in {0, 1} {
        \foreach \z in {0, 1} {
            \foreach \w in {0, 1} {
                \path[ex\y\z\w]  (n0\y\z\w) edge (n1\y\z\w);
            }
        }
    }
}%
\newcommand{\tessaractstyles}{
    \tikzset{
        hdnode/.style={},
        hdedge/.style={}
        n0000/.style={hdnode},
        n0001/.style={hdnode},
        n0010/.style={hdnode},
        n0011/.style={hdnode},
        n0100/.style={hdnode},
        n0101/.style={hdnode},
        n0110/.style={hdnode},
        n0111/.style={hdnode},
        n1000/.style={hdnode},
        n1001/.style={hdnode},
        n1010/.style={hdnode},
        n1011/.style={hdnode},
        n1100/.style={hdnode},
        n1101/.style={hdnode},
        n1110/.style={hdnode},
        n1111/.style={hdnode},
        ex000/.style={hdedge},
        ex001/.style={hdedge},
        ex010/.style={hdedge},
        ex011/.style={hdedge},
        ex100/.style={hdedge},
        ex101/.style={hdedge},
        ex110/.style={hdedge},
        ex111/.style={hdedge},
        e0x00/.style={hdedge},
        e0x01/.style={hdedge},
        e0x10/.style={hdedge},
        e0x11/.style={hdedge},
        e1x00/.style={hdedge},
        e1x01/.style={hdedge},
        e1x10/.style={hdedge},
        e1x11/.style={hdedge},
        e00x0/.style={hdedge},
        e00x1/.style={hdedge},
        e01x0/.style={hdedge},
        e01x1/.style={hdedge},
        e10x0/.style={hdedge},
        e10x1/.style={hdedge},
        e11x0/.style={hdedge},
        e11x1/.style={hdedge},
        e000x/.style={hdedge},
        e001x/.style={hdedge},
        e010x/.style={hdedge},
        e011x/.style={hdedge},
        e100x/.style={hdedge},
        e101x/.style={hdedge},
        e110x/.style={hdedge},
        e111x/.style={hdedge},
    }
}
\begin{subfigure}{.49\textwidth}
    \centering
    \begin{tikzpicture}[scale=0.7, inner sep=2pt]

        \tessaractstyles
        \tikzset{
            highlight/.style={},
            hdnode/.style={color=black!70!white},
            hdedge/.style={dotted, color=black!90!white},
            hedge/.style={},
            n0000/.style={highlight},
            n0001/.style={highlight},
            n0010/.style={highlight},
            n0011/.style={highlight},
            n0101/.style={highlight},
            n0111/.style={highlight},
            n1001/.style={highlight},
            n1100/.style={highlight},
            n1110/.style={highlight},
            n1111/.style={highlight},
            e000x/.style={hedge},
            e001x/.style={hedge},
            e00x0/.style={hedge},
            e00x1/.style={hedge},
            e01x1/.style={hedge},
            e0x01/.style={hedge},
            e0x11/.style={hedge},
            e111x/.style={hedge},
            e11x0/.style={hedge},
            ex001/.style={hedge},
            ex111/.style={hedge},
        }

        \tessaract

        \path (n1111.base) node[anchor=base] {\scriptsize$\phantom{\top = 1111} = \top$};
        \path (n0000.base) node[anchor=base] {\scriptsize$\bot = \phantom{0000 = \bot}$};

    \end{tikzpicture}
    \caption{Hasse diagram of $B^4$ and the induced subgraph for a sequence set $X$}
	\label{fig:tesseract_a}
\end{subfigure}
\begin{subfigure}{.49\textwidth}
    \centering
    \makeatletter
    \begin{tikzpicture}[scale=0.7, inner sep=2pt]

        \tessaractstyles
        \tikzset{
            highlight/.style={},
            fwdonly/.style={color=red},
            bwdonly/.style={color=blue},
            hdnode/.style={execute at begin node={\pgfsys@begininvisible}, execute at end node={\pgfsys@endinvisible}},
            n0000/.style={highlight},
            n0001/.style={highlight},
            n0010/.style={highlight},
            n0011/.style={highlight},
            n0101/.style={highlight},
            n0111/.style={highlight},
            n1001/.style={bwdonly},
            n1100/.style={fwdonly},
            n1110/.style={fwdonly},
            n1111/.style={highlight},
            hdedge/.style={transparent},
            both/.style={},
            fwd/.style={draw=red,dashed},
            bwd/.style={draw=blue,dotted},
            ex111/.style={both},
            e111x/.style={fwd},
            e11x0/.style={fwd},
            e0x11/.style={both},
            e00x1/.style={both},
            e001x/.style={both},
            e000x/.style={both},
            e00x0/.style={both},
            e01x1/.style={both},
            e0x01/.style={both},
            ex001/.style={bwd},
        }
        \tessaract

        \path (n1111.base) node[anchor=base] {\scriptsize$\phantom{\top = 1111} = \top$};
        \path (n0000.base) node[anchor=base] {\scriptsize$\bot = \phantom{0000 = \bot}$};

    \end{tikzpicture}
    \makeatother
    \caption{Induced subgraph for $X^\circ$ with removed vertices and edges.}
    \label{fig:tesseract_b}
\end{subfigure}
	\caption{Efficient computation of the largest contained well-behaved sequence set}
	\label{fig:tesseract}
\end{figure}

The Hasse diagram of $B^n$ is the $n$-dimensional hypercube where $\bot$ and $\top$ are two vertices of maximal distance and all edges point to the vertex with a smaller distance to $\top$ (see \autoref{fig:tesseract_a}).
The edges are the ordered pairs of sequences $x, y \in B^n$ with a Hamming distance $\Delta(x, y) = 1$ and $\Delta(x, \bot) < \Delta(y, \bot)$.

We will now represent a given sequence set $X \subseteq B^n$ as the subgraph of $B^n$'s Hasse diagram induced by the vertices $X$ (see \autoref{fig:tesseract_a}).
For every contained threshold set $\mathcal T_n \in X$ this representation yields a path from $\bot$ to $\top$.
Hence a sequence $x \in X$ is part of a threshold set $T_y$ with $x \in T_y \subseteq X$ if and only if in the graph for $X$, $x$ is reachable from $\bot$ and $\top$ is reachable from $x$.

The sequences for which this is true can be found by performing two depth first searches, a forward search starting from $\top$ and a backwards search starting from $\bot$.
This gives us an efficient procedure for computing $X^\circ$.

In the example of \autoref{fig:tesseract_b} the nodes and edges not visited by a forward search are drawn in red and the nodes and edges not visited by a backwards search are drawn in blue.
The black nodes are visited by both searches and correspond to the sequences of $X^\circ$.

As a further optimization for computing $(X/i)^\circ$, the pruning itself can be performed as part of the forward search, by starting the forward search at $\ohot in$ instead of $\bot$. This will stop the search from visiting vertices corresponding to sequences that would have been excluded during the pruning.
Additionally, if $X$ was already well-behaved before pruning, the backwards search will visit all vertices and thus is redundant and can be skipped.
The well-behaved pruned sequence set then contains $x/i$ for every visited vertex $x$.

\section{Certificates} \label{sec:certificates}

Even if we are convinced that our algorithm is correct, there is ample opportunity to make mistakes when implementing it.
Especially a multithreaded implementation, as we are using, is prone to having bugs that are difficult to debug and might very well go unnoticed.

Thus, to instill confidence in the algorithm's result, we would like to use formal methods that allow us to reduce the amount of code we need to trust.
The combination of all software, specifications, proofs, etc.\ that have to be manually checked to fully verify the result is called the \emph{trusted base}.
Our goal is to minimize the size of the trusted base, while preferably including general-purpose tools with established trust over problem-specific parts.
We will explicitly enumerate the trusted base at the end of \autoref{sec:verify}.

A direct approach would be to use an established proof assistant to write a machine-checked formal correctness proof of our implementation.
Currently, I do not consider this approach to be practical.

Instead we are going to augment our algorithm such that it produces some kind of \emph{certificate} of the result's correctness.
This certificate can then be checked by a much simpler, formally verified algorithm which will use our certificate-producing algorithm as an \emph{untrusted oracle}.

This approach of certifying the results of an untrusted oracle was already used by \textcite{cruz-filipeOptimizingCertifiedProof2015} to verify the computation of $s(9)$ performed in \cite{codishSortingNineInputs2016}.

For this approach to be effective we need to find a form of certificate for which there is a simple and relatively efficient verification algorithm, so that producing a machine-checked correctness proof for the implementation of the certificate-verifying algorithm becomes practical.

Note that it is sufficient to certify that our result for $s(n)$ is a lower bound, as an upper bound can be certified by exhibiting a single sorting network of that size.
The results computed using our algorithm match the sizes of long known sorting networks (see \autoref{sec:intro}).

To produce certificates for a lower bound on $s(n)$ we will introduce a formal system with a small number of inference rules based on our results for partial sorting networks from \autoref{sec:partialnets}.
A certificate will then consist of a valid derivation of a lower bound of $s(n)$ within that system.

While a direct way to generate such a derivation from our algorithm would be possible in theory, it would require prohibitive amounts of additional memory during the algorithms execution.
Instead we will present a way to construct a certificate based on the algorithm's final state.
This requires a significant amount of additional processing, but does not increase the peak memory requirements and also results in a much smaller certificate which is also proportionally faster to verify.

\subsection{Formal System for Certificates}

We define our formal system for certificates using a small number of inference rules.
Every derivable statement in our system has the form $s(X) \ge k$ for some fixed $k$ and fixed sequence set $X \subseteq B^n$ for some $n$.
The premises of inference rules can be other such statements as well as additional conditions on the occurring sequence sets.

To show that $s(n) \ge k$ it is sufficient to derive $s(X) \ge k$ for an arbitrary $X \subseteq B^n$, as $s(X)$ is monotone (\autoref{fact:mono}).

Note that the choice of inference rules here is driven by implementation and verification concerns regarding certificate generation and checking.

\begin{definition}
    Our formal system for certifying lower bounds on the minimal size of partial sorting networks consists of the following four inference rules:
{
\addtolength{\jot}{1em}
\begin{gather*}
    \inferrule*[right=Triv]{\strut}{s(X) \ge 0} \quad
    \inferrule*[right=PH]{X \subseteq B^n \\ |X| \ge n + 1}{s(X) \ge 1}\\
    \inferrule*[right=Succ]{
        X \subseteq B^n \\
        \mathcal U(X) = \{Z_1, \ldots, Z_m\} \\
        c_1, \ldots, c_m \in S_n\ang{\mathbf b} \\
        Y_1 \subseteq Z_1^{c_1} \;\;\cdots\;\; Y_m \subseteq Z_m^{c_m} \\
        s(Y_1) \ge k_1 \;\;\cdots\;\; s(Y_m) \ge k_m \\
    }{
        s(X) \ge 1 + \min\,\{k_1, \ldots, k_m \}
    }\\
    \inferrule*[right=Huffman]{
        X \subseteq B^{n + 1} \\
        b \in \ang{\mathbf b} \\
        p(X^b) = \{p_1, \ldots, p_m\} \\
        c_1, \ldots, c_m \in S_n\ang{\mathbf b} \\
        Y_1 \subseteq (X^b/p_1)^{c_1} \;\;\cdots\;\; Y_m \subseteq (X^b/p_m)^{c_m} \\
        s(Y_1) \ge k_1 \;\;\cdots\;\; s(Y_m) \ge k_m \\
    }{
        s(X) \ge H_{1{+}\mathrm{max}}\,\{k_1, \ldots, k_m\}_\#
    }\\
\end{gather*}
}
\end{definition}

\begin{lemma} The inference rules of our formal system for certifying lower bounds are sound.
\end{lemma}

\begin{proof}
The \textsc{Triv} rule holds trivially and \textsc{PH} is the Pigeon Hole Bound (\autoref{fact:nonempty}).

The \textsc{Succ} and \textsc{Huffman} rules combine the monotonicity of $s(X)$ (\autoref{fact:mono}) with the successor recurrence (\autoref{fact:succ}) and the Huffman bound (\autoref{fact:huffmanbound}), respectively.
As the premises are only lower bounds, they also derive a lower bound which is valid due to the monotonicity of $\min$ and of $H_{1{+}\mathrm{max}}$ (\autoref{fact:huffmanmono}).
Both rules also allow using permuted and optionally negated sequence sets in the premises, which is permitted as $s(X)$ is invariant under similarity (\autoref{fact:sim}).
\end{proof}

Using explicit permutations and negations as part of the inference rules allows us to use canonicalization while we produce a derivation, without introducing the notion of canonicity into the formal system.
This also means that a derivation can be verified without the need of canonicalization procedure.
Performing canonicalization is a lot more complex and also slower than the operations necessary to check the inference rules we chose.

A certificate consists of a derivation in this formal system.
We represent such a derivation as a DAG where each node is a derivation step, i.e.\ an instantiation of an inference rule.
The edges connect derived facts to their use in premises of other nodes.
To store such a certificate we fix a topological order of the derivation steps.
Then for each step, in order, we store which inference rule was used as well as information that fixes all variables appearing in that rule.
For each derived fact in the premises of the derivation step, we store only the position of that step in the topological order we chose.

Verifying such a certificate then consists of iterating over all these steps.
For each step we check that the specified variable values satisfy the rule's conditions, that the recorded positions of the derived facts used in the premises come before the position of the step they are used in and finally that each step at the recorded position for each premise derives the fact expected as premise by the step we are currently checking.

\subsection{Generating Certificates}

All lower bounds $\min b(X)$ derived during the algorithms execution can be justified using the rules of this inference system.
A direct approach would be to emit corresponding derivation steps whenever a $\min b(X)$ value changes.
The storage requirements as well as the slow down from emitting all these steps to storage makes this impractical.
An alternative would be to only maintain the derivations of the current $\min b(X)$, discarding the previous derivations of weaker facts.
This would result in a smaller certificate, but the derivation would have to be kept in memory, so that it can be updated, making the required amount of memory impractical.

Instead we are going to generate a certificate from the final set of $X \mapsto \min b(X)$ pairs.
This requires us to re-derive these lower bounds, which takes a significant amount of processing, but it also allows us to find better derivations to minimize the size and the verification runtime of the resulting certificate.

To minimize the certificate we use essentially the same strategy used in the \proc{Prune} algorithm of the generate-and-prune approach that was used to initially compute $s(9)$ \cite{codishSortingNineInputs2016}.

To do that we take all sequence sets of the pairs $X \mapsto \min b(X)$ and partition them into parts $\mathcal X_{n, k}$ containing the pairs with $X \subseteq B^n$ and $\min b(X) = k$, i.e.\ pairs with sequence sets of the same length and the same lower bound.

Now given any two subsuming $X, Y \in \mathcal X_{n, k}$ with $X \subs Y$, for every hypothetical derivation step that uses $s(Y) \ge \min b(Y) = k$ as premise, we can derive the same fact with a step that uses $s(X) \ge \min b(X) = k$ instead.
We might have to adjust the permutation and optional negation within the derivation step, which is always possible for both the \textsc{Succ} and the \textsc{Huffman} rule.

Therefore, when there is a derivation DAG using only the facts $s(X) \ge \min b(X)$ for all pairs $X \mapsto \min b(X)$, there is also a derivation DAG that uses only the subset of these facts that correspond to the non-subsumed elements of each part $\mathcal X_{n, k}$, which we will denote by $\mathcal X^{\min}_{n, k}$.

Finding the non-subsumed elements can be done sequentially for each part.
This allows us to use more elaborate data structures, that have a higher memory overhead, without going above the peak memory requirements of the initial computation.
We describe our approach below in \autoref{subsec:subsume}.

After having computed each part $\mathcal X^{\min}_{n, k}$ of non-subsumed sequence sets, we then need to find a derivation for some $s(X) \ge s(n)$ where $X \in \mathcal X^{\min}_{n, s(n)}$.
We do this using a procedure $f$ that given a sequence set $X \subseteq B^n$ selects a sequence set $f(X) \subs X$ with $f(X) \in \mathcal X^{\min}_{n, k}$ such that $k$ is maximized, but otherwise arbitrarily.
We will only compute this for $X$ where such a $f(X)$ exists.
This can be implemented using the same index data structure we use to find the non-subsumed sequence sets.

With this we can use the algorithm below to derive $s(X) \ge s(n)$ for an $X \in \mathcal X^{\min}_{n, s(n)}$.
It computes $f(B^n)$ to find such a suitable $X$ and then invokes a recursive procedure \proc{GenerateCertificateStep}$(f(B^n))$ that will derive $s(f(B^n)) \ge \min b(f(B^n)) = s(n)$.

The \proc{GenerateCertificateStep}$(X)$ procedure first checks whether the bound $s(X) \ge \min b(X)$ was already derived, in which case there is nothing to do, and otherwise tries using a \proc{Triv}, a \proc{PH}, a \proc{Huffman} and a \proc{Succ} step in this order.
For that it allows all bounds $s(Y) \ge \min b(Y)$, including subsumed $Y$, as potential premises.
For the first succeeding step, all premises $s(Y) \ge \min b(Y)$ are then replaced by $s(f(Y)) \ge \min b(f(Y)) = \min b(Y)$ so that only non-subsumed sequence sets occur.
Then \proc{GenerateCertificateStep} is invoked recursively for each of these $f(Y)$ and finally the step to derive $s(X) \ge \min b(X)$ is emitted.

\begin{algo}
\algorithm{
    global $\mathcal G \gets \emptyset$ \cmnt{sequence sets $X$ for which we derived $s(X) \ge \min b(X)$}

    \AW

    procedure $\proc{GenerateCertificate}(n)\colon$ \AC{
        $\proc{GenerateCertificateStep}(f(B^n))$
    }

    \AW

    procedure $\proc{GenerateCertificateStep}(X)\colon$ \ACX{
        if $X \in \mathcal G\colon$ return

        $\mathcal G \gets \mathcal G \cup \{X\}$

        if $\min b(X) \le 1\colon$ \AC{
            \cmnt{$\downarrow$ Sufficient for well-behaved sequence sets}

            emit \proc{Triv} or \proc{PH} step for $s(X) \ge \min b(X)$

            return
        }

    }

    $\hskip0.5em\;\vdots$
}

\algorithm{
    $\hskip0.5em\;\vdots$
    \AXC{

        for $Y$ in $X^\ang{\mathbf b}\colon$ \cmnt{First try deriving via the Huffman bound} \AC{
            if $\min b(X) \le H_{1{+}\mathrm{max}}\, \{\, \min b((Y/i)^\sim) \mid i \in p(Y) \,\}_\#\colon$ \AC{
                $\mathcal Z \gets \{\, f(Y/i) \mid i \in p(Y) \,\}_\#$

                for $Z \in \mathcal Z\colon$ \AC {
                    $\proc{GenerateCertificateStep}(Z)$
                }

                emit \proc{Huffman} step for $s(X) \ge \min b(X)$ using $s(Z) \ge \min b(Z)$ for all $Z \in_\# \mathcal Z$

                return
            }
        }

        \cmnt{$\downarrow$ Now only \proc{Succ} remains, and thus it can derive $s(X) \ge \min b(X)$}

        $\mathcal Z \gets \{\, f(Y) \mid Y \in \mathcal U(X) \,\}_\#$

        for $Z \in \mathcal Z\colon$ \AC {
            $\proc{GenerateCertificateStep}(Z)$
        }

        emit \proc{Succ} step for $s(X) \ge \min b(X)$ using $s(Z) \ge \min b(Z)$ for all $Z \in_\# \mathcal Z$
    }
}
\end{algo}

Additionally as a small optimization of the certificate size, we can remove all nodes that correspond to a step using the \textsc{Triv} or \textsc{PH} rule.
Whenever another step refers to them, instead of storing a reference to the now removed step, we store whether \textsc{Triv} or \textsc{PH} was used.
For these premises, we also do not store any permutation or indicate whether optional negation was used.
Instead we assume that the premise uses the largest sequence set admissible by the inference rule, requiring no permutation or negation.

This is possible since both \textsc{Triv} or \textsc{PH} stay valid when replacing the sequence set $X$ in the conclusion with a sequence set that is subsumed by $X$.

We omit a correctness proof of this certificate generation procedure, as our intention is to base the trust on formally verified checking of the resulting certificate.

\subsection{Subsumption Testing} \label{subsec:subsume}

When performing subsumption testing, we can handle negation by performing two separate tests for subsumption that only consider permutations.
Therefore, for the rest of this section, we use \emph{subsumption} to refer to subsumption restricted to permutation-similarity.

To efficiently test for subsumption, we use the general idea presented in \cite{codishSortingNineInputs2016} which defines several functions from sequence sets to the integers, all of them order-preserving with respect to subsumption.
We will call these functions \emph{subsumption abstractions}.
This allows quickly detecting non-subsumption for a large number of required subsumption tests, leaving only a small fraction for which a more expensive test is required. We use a larger set of such functions than used in \cite{codishSortingNineInputs2016} to exclude more non-subsuming pairs.
We refer to our implementation in \cite{harderJixSortnetopt2020} for more details.

If a candidate pair $X, Y$ passes the tests using subsumption abstractions, we need to search for a permutation under which one includes the other, i.e.\ $X \subseteq Y^\sigma$.
To avoid testing all $n!$ permutations, we use the approach described by \textcite{frasinaruImprovedSubsumptionTesting2019}.

There a \emph{compatibility test} is defined for each ordered pair of channels $i, j$.
The compatibility test must succeeds for those $i, j$ for which there exists a permutation $\sigma$ mapping $i$ to $j$ with $X \subseteq Y^\sigma$.
For $i, j$ without such a permutation the result can be arbitrary, i.e.\ the compatibility test is an over-approximation of this property.
Each failing compatibility test excludes several permutations that could witness $X$ subsuming $Y$.
The subsumption test is then completed by enumerating and testing the remaining permutations.

To efficiently enumerate the candidate permutations they define a \emph{subsumption graph}.
The subsumption graph is a bipartite graph and has the channels of $X$ as vertices on one side and the channels of $Y$ as vertices on the other side.
There is an edge between channel $i$ of $X$ and channel $j$ of $Y$ exactly if $i$ and $j$ pass the compatibility test. The candidate permutations then correspond to the perfect matchings of this subsumption graph.

We use a different strategy than \cite{frasinaruImprovedSubsumptionTesting2019} to define the compatibility test.
Given a subsumption abstraction $f$, we can define a family of functions $f_i(X) = f(X/i)$, one for each channel $i \in [n]$.
We call such a family of functions $f_i$ a \emph{channel abstraction}.

Given a channel abstraction $f_i$, there can be a permutation $\sigma$ with $X \subseteq Y^\sigma$ only if for all channels $i \in [n]$ we have $f_i(X) \le f_j(Y)$ where $j = i^\sigma$.
This makes $f_i(X) \le f_j(Y)$ a compatibility test.
When building a subsumption graph, we use the conjunction of several such tests obtained from different channel abstractions.

As a significant further improvement, we store the sequence sets in a spatial search tree that maintains bounding hyperrectangles for each subtree.
Each sequence set is stored at the point in space given by the vector containing all channel abstractions of that sequence set. When recursing the search tree, looking for a subsuming (or subsumed) candidate, we perform an approximate test on whether a perfect matching exists, using the minimal (or maximal) coordinates of the bounding box as channel abstractions. If no such matching exists, we can skip the complete subtree contained in that bounding box, otherwise we recurse, performing the same check for all inner nodes.

\section{Formal Verification} \label{sec:verify}

Having generated a certificate, we wish to check it using a program that has a formal correctness proof.
Our certificate checker and a corresponding correctness proof are written and verified using the Isabelle/HOL proof assistant \cite{nipkowIsabelleHOLProof2002}.
Having a correctness proof that passes verification by the Isabelle/HOL proof assistant means that only the formalized problem statement, given below, has to be manually verified.

Note that for formal verification, we do not assume correctness of any proofs presented in this paper.
Every result which we require is proved and machine checked as part of the checker's formal correctness proof.

To perform a certificate check, we use Isabelle/HOL's code generation feature \cite{haftmannCodeGenerationHigherOrder2010} which allows extracting a function defined in Isabelle/HOL into a semantically equivalent function in one of several general purpose functional programming languages.
In our case we extract to the Haskell language.

The source code for the checker and its correctness proof are available at \cite{harderJixSortnetopt2020}.

\newcommand{\isaqopen}{\text{\guilsinglleft}}
\newcommand{\isaqclose}{\text{\guilsinglright}}
\newcommand{\isaq}[1]{\isaqopen#1\isaqclose}

\newcommand{\isaassign}{\mathrel{\text{\raise.165ex\hbox{$\scriptstyle:$}}\!\!=}}
\newcommand{\isadcolon}{\mathrel{\text{\raise.1ex\hbox{$\scriptstyle::$}}}}

For this section only, we number the channels starting at $0$, to match Isabelle/HOL's \textit{nat} type.
Also, Isabelle/HOL uses \textit{False} and \textit{True} as Boolean values, where we used $0$ and $1$.

To formally state the problem of certifying the minimal size of an $n$-channel sorting network in Isabelle/HOL,
we first define infinite Boolean sequences, represented as a function from the naturals to the Booleans:

\vskip0.5em\noindent
$\textbf{type-synonym }\textit{bseq} = \isaq{\textit{nat} \Rightarrow \textit{bool}}$
\vskip0.5em

We then represent finite sequences as infinite sequences that are constant True for all values starting at the index equal to their fixed length.

\vskip0.5em\noindent
$\textbf{definition }\textit{fixed-len-bseq} \isadcolon \isaq{\textit{nat} \Rightarrow \textit{bseq} \Rightarrow {bool}} \textbf{ where}$ \\
$\strut\quad\isaq{\textit{fixed-len-bseq n x} = (\forall i \ge n.\ x\ i = \textit{True})}$
\vskip0.5em

We represent a comparator as a pair of channel indices.
Here, we do not require both channels to be distinct.
\vskip0.5em\noindent
$\textbf{type-synonym }\textit{cmp} = \isaq{\textit{nat} \times \textit{nat}}$
\vskip0.5em

Next we define the action of a comparator on an infinite sequence:
\vskip0.5em\noindent
$\textbf{definition }\textit{apply-cmp} \isadcolon \isaq{\textit{cmp} \Rightarrow \textit{bseq} \Rightarrow {bseq}} \textbf{ where}$ \\
$\strut\quad\isaq{\textit{apply-cmp c x} = (\textrm{let } (i, j) = c \textrm{ in } x(
    i \isaassign \textit{min}\ (x\ i)\ (x\ j),\,
    j \isaassign \textit{max}\ (x\ i)\ (x\ j)
    ))}$
\vskip0.5em
The syntax $f(i \isaassign a)$ stands for the function $g$ defined by $g\ i = a$ and $g\ k = f\ k$ for all $k \ne i$.

With that we are ready to define lower size bounds for a partial sorting networks. The following definition allows us to state $s(X) \ge k$ as $\textit{partial-lower-size-bound }X'\ k$ where $X'$ contains the sequences of $X$, represented as \textit{bseq} values:

\vskip0.5em\noindent
$\textbf{definition }\textit{partial-lower-size-bound} \isadcolon \isaq{\textit{bseq set} \Rightarrow \textit{nat} \Rightarrow {bool}} \textbf{ where}$ \\
$\strut\quad\isaqopen\textit{partial-lower-size-bound X k} = ($ \\
$\strut\quad\quad\forall \textit{cn}.\ (\forall x \in X.\ \textit{mono}\ (\textit{fold apply-cmp cn x})) \rightarrow \textit{length cn} \ge k$ \\
$\strut\quad)\isaqclose$
\vskip0.5em

Here $\textit{cn}$ ranges over all comparator networks having no exchanges, so that $\textit{length cn}$ is the size of $\textit{cn}$.
With $\textit{fold apply-cmp cn x}$ we apply that network to a sequence $x$ and $\textit{mono y}$ is true exactly if $y$ is sorted.

Note that we, again, do not exclude comparators where both channels are equal. Since the action we defined for such a comparator is the identity, this does not affect the minimal number of required comparators.

Completing our formalization of sorting network size bounds, we can use this to define a property equivalent to $s(n) \ge k$ which is the same as $s(B^n) \ge k$:

\vskip0.5em\noindent
$\textbf{definition }\textit{lower-size-bound} \isadcolon \isaq{\textit{nat} \Rightarrow \textit{nat} \Rightarrow {bool}} \textbf{ where}$ \\
$\strut\quad\isaq{\textit{lower-size-bound n b} = \textit{partial-lower-size-bound } \{x.\ \textit{fixed-len-bseq n x}\}\ k}$
\vskip0.5em

We do not exclude comparators with indices beyond the first $n$ channels.
For all input sequences, these channels are constant \textit{True}, and thus the action of such a comparator would be the identity or an exchange and hence this does not affect the minimal number of required comparators.

We then define a certificate checking function \textit{check-proof-get-bound}.
This implements the checking as described in the previous section.
The argument of this function is the certificate and the result is either $\textit{None}$ if the certificate could not be validated or $\textit{Some }(n,\, k)$ if the certificate certifies $s(n) \ge k$.

The certificate is represented using algebraic data types, integers and functions between these. In that way the sets of representable certificate values for the corresponding Haskell and Isabelle/HOL type are the same.
As Isabelle/HOL does not have empty types and does not require functions to be effectively computable, diverging Haskell terms are semantically equivalent to Isabelle/HOL's \textit{undefined} term, and do not pose a problem here.

We are going to omit further details about the exact representation of certificates as well as the implementation details of \textit{check-proof-get-bound}.
The interested reader can find those details in the source code repository \cite{harderJixSortnetopt2020}.

Finally we want to show correctness of this certificate checking function.

\begin{lemma} \label{fact:verifycorrect}
    When there exists a certificate for which the Isabelle/HOL function \textit{check-proof-get-bound-spec} returns a value of $\textit{Some } (n,\, k)$, the bound $s(n) \ge k$ holds.
\end{lemma}

\begin{proof}
We can express this using the following equivalent formal statement:

\vskip0.5em\noindent
$\textbf{lemma }\textit{check-proof-get-bound-spec}\colon$\\
$\strut\quad\textbf{assumes }\isaq{\textit{check-proof-get-bound cert} = \textit{Some }(\textit{width},\, \textit{bound})}$\\
$\strut\quad\textbf{shows }\isaq{\textit{lower-size-bound } (\textit{nat width})\ (\textit{nat bound})}$
\vskip0.5em

We prove this by formalizing most of \autoref{sec:partialnets} in Isabelle/HOL.
This includes a formalization of Huffman's algorithm on arbitrary Huffman algebras and corresponding correctness proofs.
With that we can formally show correctness of the rules defining the certificate's formal system.
Then we show that \textit{check-proof-get-bound-spec} verifies that each individual step of the certificate is a valid instantiation of such a rule and that each step only references steps that precede it.
Given this, we can finish the proof by complete induction on the sequence of steps.

As the corresponding formal proof is machine checked by the Isabelle/HOL system, we omit further details.
We again refer the interested reader to \cite{harderJixSortnetopt2020}.
\end{proof}

\subsection{Extraction to Safe Haskell}

To run the certificate checking, we automatically extract a Haskell implementation of \textit{check-proof-get-bound-spec} using Isabelle/HOL's code generation features \cite{haftmannCodeGenerationHigherOrder2010}.

The code extraction maps certain Isabelle/HOL types to their corresponding types in the Haskell base library.
This includes the \textit{option} type and its \textit{Some} constructor which are mapped to Haskell's \textit{Maybe} type and its \textit{Just} constructor.
Thus we expect the certificate checking to terminate with a value of $\textit{Just }(n,\, k)$ instead of $\textit{Some }(n,\, k)$.

To pass a certificate value to the extracted certificate checking function,
we use an unverified certificate file parser written in Haskell and pass the result to the extracted certificate checking function.
The parser turns a certificate file given as a sequence of bytes into a value of the certificate type whose Haskell definition is extracted from its Isabelle/HOL definition.

In general, Haskell has certain \emph{unsafe} functions available in the base library that, if improperly used, break essential invariants of the Haskell compiler and runtime system.
For example, using \textit{unsafePerformIO}, it would be possible to construct a value of a function type that when applied twice to the same value produces different results and thus is not a function.
As Isabelle/HOL assumes that values of a function type are indeed functions, we need to make sure to not construct such values as part of the certificate.

Safe Haskell is a subset of the Haskell language and base library that excludes these unsafe functions \cite{tereiSafeHaskell2012}.
When declaring a Haskell source file as Safe Haskell the Haskell compiler automatically checks that only this subset is used.
We enable Safe Haskell for both the extracted code as well as the unverified parser.

Safe Haskell does offer an escape hatch that allows declaring additional general Haskell functions as usable from Safe Haskell, but this requires an explicit annotation, making obvious the places where it is necessary to manually verify the proper use of unsafe functions.

To speed up verification, we apply two changes to the extracted code.

The first change is the addition of strictness annotations to the parameters in the definition of certain functions and data type constructors.
By default Haskell uses \emph{lazy evaluation} (also called \emph{call-by-need}), which defers the computation of an expression until its value is required by another computation or is output to the user.
A strictness annotation ensures that whenever a value of an expression using the affected definition is required, the annotated parameter's value is computed first.

Strictness annotations only affect the evaluation order, and thus will not change the result of a terminating computation.
They might turn a terminating computation into a diverging one, but as Isabelle/HOL's code extraction does not ensure that the extracted code terminates for all defined values in the first place, this does not weaken our guarantees.
The formally proved properties will still hold for every terminating computation.

The second change is that we make the \textit{par} function provided by the GHC Haskell compiler available to Safe Haskell.
This function takes two arguments and returns the second, while starting a parallel computation of the first value.
Apart from the differences in evaluation order, it semantically equivalent to the definition $\textit{par a b} = b$.
As such it unconditionally fulfills the relevant requirements of a Safe Haskell function.

To use \textit{par} from Isabelle/HOL while being able to reason about it, we define it as $\textit{par a b} = b$ within Isabelle/HOL.
We then replace the extracted definition with the GHC provided, semantically equivalent \textit{par}.
We use this \textit{par} function to enable parallel checking of the certificate's steps.

These two changes are not necessary for performing certificate checking, but as they provide a significant speed up, I have not checked a certificate for $s(11)$ without applying them.

\subsection{Trusted Base}

The computation and verification of $s(n)$ involves several steps implemented across many components.
Formal verification allows us to trust the final result to be correct without having to trust every involved component.
The subset that has to be trusted to establish correctness is called the \emph{trusted base}.

Note that we do not require, or even expect, the trusted base to be free of all bugs.
Not every bug within a single component leads to incorrectness of the overall system.
As an example consider a compiler that miscompiles some programs, but not any of the programs used for verification.

In contrast to the adversarial setting of verifying an algorithm for arbitrary inputs or establishing the security of a software system, here our goal is to verify a specific single result of a computation.
Therefore only bugs that are triggered during that verification are relevant.
This allows us to be not concerned by bugs that cause a crash, abort the program with an error or otherwise make themselves known.

In the table below, we list all components used during the computation and verification of $s(n)$ for a given $n$. We indicate whether a component is part of the trusted base. If it is, we provide some information useful for assessing its trustworthiness, otherwise we summarize why that component does not have to be trusted.

\begin{longtable}{%
    >{\raggedright\arraybackslash}p{3.7cm}%
    >{\raggedright\arraybackslash}p{1.3cm}%
    p{\dimexpr \textwidth-3.7cm-1.3cm-6\tabcolsep}%
}
    \toprule
    Component & Trusted Base & Notes \\
    \midrule
    \endfirsthead
    \toprule
    Component & Trusted Base & Notes \\
    \midrule
    \endhead
    \bottomrule
    \endfoot
    \bottomrule
    \endlastfoot
    Hardware and Operating System & Yes &
        All other used components support multiple operating systems and can run on hardware from different vendors.
        The certificate checking results reported below were obtained on an \enquote{AMD Ryzen 9 3950X} CPU with ECC memory running Linux.
        \\[0.5em]
    GHC Haskell Compiler and Runtime & Yes &
        GHC is the most widely used Haskell implementation.
        While there are other Haskell implementations, we do make optional use of GHC specific extensions for parallelizing the certificate checking and only provide build scripts that use GHC.
        \\[0.5em]
    Isabelle/HOL to Haskell Code Extraction & Yes &
        The translation from a functional program specified in Isabelle/HOL to an equivalent functional program in one of the supported target languages and arguments for the correctness of this translation are given in \cite{haftmannCodeGenerationHigherOrder2010}. We make use of the \emph{Code-Target-Numeral} feature, which maps natural numbers and integers to the corresponding native types of the target language, trusting the correctness of their implementation.
        \\[0.5em]
    Isabelle/HOL Proof Checking & Yes &
        The Isabelle system itself is constructed in a way that minimizes its trusted base by using a small and well tested logical inference kernel. Every proof that is accepted by Isabelle is ultimately reduced to the valid inferences offered by this inference kernel. That way, the implementation of more complex proof automation features does not have to be trusted.
        \\[0.5em]
    Formal Problem Statement & Yes &
        This is the only trusted component which is not general purpose, but rather specific to this work. The full formal problem statement is given above, including arguments for its correctness.
        \\[0.5em]
    Certificate Checking Implementation & No &
        The implementation of the certificate checking algorithm is not part of the trusted base, as it is covered by a formal machine checked correctness proof.
        \\[0.5em]
    Formal Certificate Checking Correctness Proof & No &
        The formal correctness proof of the certificate checking algorithm is not part of the trusted base, as it is fully machine checked using the Isabelle/HOL proof assistant. The proof does not assume any lemmas or theorems presented in this paper. As part of the correctness proof, all required facts are derived from the definitions of the formal problem statement.
        \\[0.5em]
    Search Implementation & No &
        The implementation of the algorithm described in \autoref{sec:algo} and \autoref{sec:parallel} does not have to be trusted, as it outputs data from which we generate a certificate.
        \\[0.5em]
    Search Algorithm Correctness Proof & No &
        The correctness proof given in \autoref{sec:algo} is not part of the trusted base, as the result of the algorithm is independently verified via a certificate generated from the algorithm's final state.
        \\[0.5em]
    Certificate Generation & No &
        The certificate generation does not have to be trusted, as the certificate checking makes no assumption on how the certificate was generated.
        \\[0.5em]
    Certificate Parser & No &
        During certificate checking, we use an unverified parser that parses a certificate file into a certificate datatype. It is not part of the trusted base, as the formal correctness proof of the certificate checker makes no assumptions about the well-typed value passed as certificate.
\end{longtable}

\section{Computing $s(11)$ and $s(12)$} \label{sec:compute}

Using the parallel implementation of the algorithm describedin \autoref{sec:parallel} I computed $s(11)$.
The hardware used for this was an \enquote{AMD EPYC 7401P 24-Core Processor} with 48 hardware threads running at 2\,GHz.

The computation took 4 hours and 51 minutes with a peak memory usage of 178\,GiB.
It found that there are no 11-channel sorting networks with fewer than 35 comparators and generated 93\,GiB of output data to store the partition of sequence sets used for the subsequent certificate generation.
This output data contains 2,462,890,689 sequence sets for which bounds were computed.

\begin{theorem}
	The minimal number of required comparators of a sorting network with $11$ channels is $s(11) = 35$.
\end{theorem}

\begin{proof}

	To verify this result, the first step of certificate generation consists of pruning all subsumed sequences sets in the output data.
	The parts of the partition were pruned sequentially, but each individual part was processed in parallel.
	This way only one part at a time is residing in memory.
	Overall this took 2 days 3 hours and 39 minutes with a peak memory usage of 16\,GiB.
	The resulting output consists of 15,432,816 non-subsumed sequence sets requiring 874\,MiB of storage.

	Next a certificate is generated by finding derivations of the lower bound of each of these sequence sets.
	This took 19 hours and 2 minutes with a peak memory usage of 54\,GiB.
	The resulting certificate has a storage size of 2926\,MiB and contains 12,659,079 steps.

	This certificate is available at \cite{harderCertificateMinimalSize2019} to allow independent verification of this result without recomputing a certificate.

	The certificate verification was then performed on a computer with an \enquote{AMD Ryzen 9 3950X 16-Core Processor} with 32 hardware threads running at 4.1\,GHz.
	It took 34 minutes with a peak memory usage of 6\,GiB and returned $\textit{Some}\ (11,\, 35)$, indicating a successful verification.

	From the machine checked formal correctness proof of the certificate verification routine (\autoref{fact:verifycorrect}) we get that $s(11) \ge 35$, matching the known upper bound $s(11) \le 35$.

\end{proof}

\begin{corollary}
	The minimal number of required comparators of a sorting network with $12$ channels is $s(12) = 39$.
\end{corollary}

\begin{proof}
	Via Van Voorhis's bound (\autoref{fact:vanvoorhis}) using $s(11) = 35$ and the known upper bound $s(12) \le 39$.
\end{proof}

To see how our algorithm scales, and to compare our algorithm to previous approaches, we also perform the computation of $s(n)$ for smaller $n$ where Van Voorhis's bound is not strict.

In the table below we see for how many sequence sets a non-trivial bound was computed, how many of these sequence sets are not subsumed by another sequence set of the same bound and how many of the remaining bounds are derived as part of the final certificate.

Note that these values are affected by scheduling non-determinism, and will vary slightly with each run of the algorithm.

\begin{center}
	\begin{tabular}{rrrrrr}
		\toprule
		$n$ & Sequence Sets & Non-Subsumed & Certificate Steps \\
		\midrule
		5   & 34            & 26           & 24                \\
		7   & 550           & 247          & 212               \\
		9   & 206,279       & 13,034       & 11,934            \\
		11  & 2,462,890,689 & 15,432,816   & 12,659,079        \\
		\bottomrule
	\end{tabular}
\end{center}

While we have only four data points, and thus cannot make good predictions
for larger $n$, all measured values are consistent with double
exponential or even faster growth.

The following table contains the corresponding memory usage and runtime for the initial search, the pruning of subsumed sequence sets, the certificate generation and the certificate verification.

\begin{center}

	\newcommand{\dcol}[1]{\multicolumn{4}{c}{#1}}
	\newcommand{\ccs}{\kern 0.75em}
	\newcommand{\dcs}{\kern 0.66em}
	\newcommand{\lt}{{<}\kern 0.1em}
	\newcommand{\nt}{$^\star$}

	\noindent
	\begin{tabular}{r@{\ccs}r@{\,}l@{\dcs}r@{\,}l@{\ccs}r@{\,}l@{\dcs}r@{\,}l@{\ccs}r@{\,}l@{\dcs}r@{\,}l@{\ccs}r@{\,}l@{\dcs}r@{\,}l}
		\toprule
		$n$ &        \dcol{Search}        &       \dcol{Pruning}           & \dcol{Cert. Generation}  & \dcol{Cert. Verification}  \\
		\midrule
		5  & 14  & MiB & $\lt0.1$ & s    & $\lt1$ & MiB & $\lt0.1$ & s    & $\lt1$ & MiB & $\lt0.1$ & s    & $\lt1$ & GiB & $\lt1$   & s \\
		7  & 16  & MiB & $\lt0.1$ & s    & $\lt1$ & MiB & $\lt0.1$ & s    & $\lt1$ & MiB & $\lt0.1$ & s    & $\lt1$ & GiB & $\lt1$   & s \\
		9  & 58  & MiB & 0.5      & s    & 15     & MiB & $0.6$    & s    & 54     & MiB & $0.4$    & s    & $\lt1$ & GiB & $3.1$    & s \\
		11 & 178 & GiB & 4\,h 51  & m\nt & 16     & GiB & 2\,d 5   & h\nt & 54     & GiB & 19\,h 2  & m\nt & $5.8$  & GiB & 33\,m 40 & s \\
		\bottomrule
		\multicolumn{17}{l}{\footnotesize Using an \enquote{AMD Ryzen 9 3950X 16-Core Processor} at 4.1\,GHz} \\
		\multicolumn{17}{l}{\footnotesize $^\star$Using an \enquote{AMD EPYC 7401P 24-Core Processor} at 2\,GHz}
	\end{tabular}
\end{center}

This allows us to compare our our approach to the previous state of the art.

The initial computation of $s(9)$ was done once using the generate-and-prune method and once using a hybrid of generate-and-prune and SAT solving \cite{codishSortingNineInputs2016}.
Using the generate-and-prune method, the computation took around 12 days on 144 \enquote{Intel E8400} dual core CPUs at 2\,GHz and produced around 50\,GB of data used for verification.
Using the combined approach the runtime was reduced to around 7 days.

For a fair comparison, though, we need to look at implementations that have seen a similar amount of optimization work and we also should consider improvements made to the generate-and-prune method since then.

The runtime of the generate-and-prune method is dominated by the pruning step which finds non-subsumed sequence sets.
By using the improved algorithm for the subsumption check described in \cite{frasinaruImprovedSubsumptionTesting2019}, the runtime is reduced to 29 hours on a computer with \enquote{Intel Xeon E5-2670} CPUs having a total of 32 cores running at 2.60\,GHz.

By further improving the pruning step using a subsumption testing method similar to the one used for our certificate generation (see \autoref{subsec:subsume}), the runtime of generate-and-prune for 9 channels is reduced to 44 minutes on an \enquote{AMD Ryzen 9 3950X 16-Core Processor} running at 4.1\;GHz \cite{harderJixSortnetoptgnp2020}.

Comparing this to our new algorithm, we see that it significantly improves the state of the art. The runtime for $s(9)$ is improved from $44$ minutes to a fraction of a second.
This improved performance allowed us to compute $s(11)$.

\section{Conclusion}

Using Huffman's algorithm we generalized Van Voorhis's bound to partial sorting networks.
This allowed us to construct a dynamic programming algorithm that computes the minimal size $s(n)$ required for an $n$-channel sorting network.
Using an optimized parallel implementation of our algorithm we computed $s(11) = 35$. From this we also derived $s(12) = 39$ with the help of Van Voorhis's bound.

To verify our computation we extended our algorithm to generate a certificate that can be efficiently checked.
We presented a definition of lower size bounds for sorting networks formalized in Isabelle/HOL.
Based on that we developed certificate checker with a machine checked formal correctness proof.
Using this formally verified certificate checker we successfully checked the certificate for $s(11) \ge 35$.
The upper bounds $s(11) \le 35$ and $s(12) \le 39$ come from long-known sorting networks having those sizes and number of channels.

\subsection{Future Work}

The subsumption based pruning strategy already used in \cite{codishSortingNineInputs2016} proved to be very effective at minimizing the size of the certificate generated by our algorithm.
In principle an algorithm using subsumption during the search could directly discover the derivation contained in the certificate of $s(11) \ge 35$, performing much less work than our algorithm does.
In practice our algorithm seems to derive intermediate bounds in an order that does not work well with on-line subsumption.
Analyzing this further requires more experiments.
Finding a way to make on-line use of subsumption has the potential to significantly decrease the memory requirements, which are the biggest obstacle to computing $s(n)$ for larger $n$.

In \cite{codishSortingNineInputs2016}, performing parts of the search with SAT solvers resulted in a significant speed up compared to the initial generate-and-prune implementation. It would be interesting to find a CNF encoding of the minimal size (partial) sorting network problem that uses and benefits from the Huffman bound.

\section*{Acknowledgments}

I would like to thank Maja Kądziołka, Paul Khuong and Erika (@rrika9) for their valuable suggestions and comments on draft versions of this paper.

\printbibliography

\appendix

\section{Parallel Implementation} \label{sec:parallel}

In \autoref{sec:compute}, to reduce the cost of computing $s(11)$, we use an optimized parallel implementation of the algorithm described in \autoref{sec:algo}.
Here, we document some implementation choices made, including how the algorithm was parallelized and which heuristics were used.
We do not conduct a detailed performance analysis nor claim that these choices are optimal in any way.

\subsection{Representing Sequence Sets}

Storing the $X \mapsto b(X)$ pairs whenever $b(X) \ne b_0(X)$ takes up the vast majority of memory used by the algorithm.
Thus it is important that we use a memory efficient representation for these sequence sets.
We represent a sequence set $X \subseteq B^n$ as a bit-vector of length $2^n$ where the value at index $i$ is set exactly if the length $n$ binary representation of $i$ is in $X$.
For $X \mapsto b(X)$, we store this bit-vector using a packed representation where all $8$ bits of each byte are used to store elements of the bit-vector.

To speed up computations on sequence sets, we unpack these bit-vectors to a representation using one byte per element before performing any operations on them, as addressing whole bytes is often faster than addressing individual bits.

\subsection{Heuristic Choices for Sufficiently Improved Bounds}

Our algorithm includes two different kinds of heuristic choices.
One is the selection of subcomputations to recurse on between bound updates.
We will deal with this below when we describe how to parallelize the algorithm.
The other heuristic decides how much a bound has to improve before the \proc{Improve} routine returns.

The choice has to be somewhere between returning as soon as the bounding interval of the target sequence set changed and continuing until the interval becomes singleton when the target sequence set is fathomed.

After trying several variations, I settled on the heuristic that keeps improving bounds until either the lower bound of the interval improved or it becomes singleton, ignoring other improvements of the interval's upper bound.
While this seems to perform well in practice, a comprehensive comparison of different strategies would require further experiments.

\subsection{Opportunities for Parallelization}

A divide-and-conquer algorithm that splits a problem into independent subproblems can be parallelized by performing the corresponding subcomputations in parallel using multiple threads.
Adding dynamic programming via memoization to the picture only slightly complicates this.
A data structure that allows for concurrent reads and updates is required for the table caching the results of already computed subproblems.
It is also advisable to allow locking of individual entries of that table.
When a subcomputation is spawned, the corresponding entry is locked so that a concurrent request for the same subproblem will wait for that computation to finish instead of spawning a concurrent redundant computation.

In our case, to make efficient use of successive approximation, this strategy alone will not sufficiently parallelize the computation.
To see why that is the case, we need to take a closer look at how the improvement routines recurse.

For both, improvement via the Huffman bound and improvement via successors, we repeatedly recurse on some non-fathomed sequence sets until progress is made or until no further progress is possible.
So far we have not specified how to select which of these sequence sets the improvement routines should recurse on.
While the order does not matter for the correctness of the algorithm it can have a big impact on the runtime and number of encountered subproblems.

For both these improvement strategies, it is possible that recursing on a single subproblem is sufficient for making progress.
As we are not always able to detect this in advance, though, we have to make a heuristic decision.
Nevertheless, in a non-parallel setting it would be always advantageous to recurse only on a single subproblem at a time, repeatedly checking for progress when that subcomputation finishes and returning as early as possible.
However, in the parallel setting this is not a good approach at all, as it completely serializes all subcomputations.

To get most of the advantages of this non-parallel approach in the parallel setting, we use two approaches.

The first is to asynchronously check for progress using the bound update routines after every single subcomputation finishes, even if other subcomputations are still in progress concurrently.
If progress was made, the \proc{Improve} routine can return immediately and stop all other subcomputations it spawned whether directly or indirectly via further recursion.

This allows running subcomputations in parallel, while reducing the amount of work performed for subproblems whose result is not needed for progress at that moment.
Even if some amount of work is performed for these unnecessary subproblems, it might not all be wasted, as all improved bounds discovered during this work are still cached and might be useful in future computations.

While this slightly reduces the downsides of running multiple subcomputations of a single \proc{Improve} invocation in parallel, it is still not a good idea to to eagerly spawn all of them in parallel.

To address this, we use a second approach inspired by work-stealing \cite{blumofeSchedulingMultithreadedComputations1999} and the spark concept of the Glasgow Haskell Compiler's runtime system \cite{marlowRuntimeSupportMulticore2009}.

In work-stealing, subcomputations are kept in a queue local to each hardware thread and only migrated to other hardware threads when these become idle and have no subcomputations left in their own queue.
The goal there is to reduce the involved communications overhead when migrating subcomputations to different hardware threads to balancing work across all of them.

The Glasgow Haskell Compiler allows parallelizing Haskell programs by specifying places in the program where opportunistic parallel evaluation of a given subcomputation can be beneficial.
During runtime whenever such a place is encountered, a \emph{spark} is created.
This does not cause a parallel subcomputation to be started immediately, instead, whenever a hardware thread would become idle, starting the computation of a spark is among the choices the runtime's scheduler can make.
For more details about the parallelization features of the runtime system see \cite{marlowRuntimeSupportMulticore2009}.

As we are not much concerned about the communication overhead of migrating subcomputations between hardware threads, we can simplify our scheduler by using two global queues instead of having queues for each hardware thread.
We have an \emph{active queue} and a \emph{pending queue}, which together partition all prepared subcomputations by whether they have been started yet.
Subcomputations in our \emph{pending queue} are somewhat similar to GHC's sparks.

Whenever a hardware thread would become idle, it first continues work on a subcomputation in the active queue that was already started and is not blocked waiting for another subcomputation.
Only if no such subcomputation exists, it starts a new subcomputation, moving it from the pending to the active queue.
We call this \emph{deferred spawning}.

This allows the \proc{Improve} routine to prepare and enqueue all subcomputations, while only explicitly starting a single subcomputations at a time.
That way additional subcomputations are only started when necessary to prevent a hardware thread from becoming idle.
If the \proc{Improve} routine makes progress before that happens, we detect that using our asynchronous progress checks and then cancel all the subcomputations the routine prepared so that they are not started at all, should they still be in the pending queue.

In some cases, when performing improvement via successors, we do know in advance that multiple subcomputations are required for progress. In that case we also allow the \proc{Improve} routine to start all of them eagerly, to avoid a hardware thread becoming idle and starting a different possibly unnecessary subcomputation.

The active queue is implemented as a FIFO queue, giving roughly equal priority to all started subcomputations, while the pending queue is a priority queue.
This allows the \proc{Improve} routine to assign a heuristic priority to all prepared subcomputations, which further helps in starting subcomputations that are likely to be useful for making progress.

\subsection{Threading Runtime}

The parallelization strategies described in the previous subsection require primitives (e.g.\ deferred spawning, manually starting deferred subcomputations or recursive cancellation) that are usually not provided by an operating system's or programming language runtime's multithreading facilities.

For the implementation, I decided to realize these primitives within the Rust programming language \cite{rustprojectdevelopersRustProgrammingLanguagen.d.,matsakisRustLanguage2014}, making use of the recently added async-await feature \cite{matsakisRustBlogAsyncawait}.
Similar features are provided by several programming languages, often intended to facilitate writing programs that make efficient use of non-blocking I/O-operations.

In addition to non-blocking I/O, Rust's async-await implementation is also well suited for writing problem specific multithreaded task schedulers.
It offers a first-class representation of concurrent subcomputations called \emph{futures}.
In Rust they provide the necessary freedom required to implement the primitives described below.

\paragraph{Deferred Spawning}

As described in the previous section, by default threads spawned to perform subcomputations should only be started when otherwise a hardware thread would become idle.
We write \inlinealg{\,spawn$\colon$ $\{\cdots\}$\,}\ to perform such a deferred spawn of a thread performing the enclosed subcomputation.

When it becomes necessary to start a thread that was deferred, the thread with the smallest priority value among all deferred threads is selected.
That priority value $x$ can be specified when spawning a thread using \inlinealg{\,spawn$\colon$ $\{$global priority $x;$ $\cdots\}$\,}.
Besides the \textbf{global priority} threads also have a \textbf{group priority} which is specified in the same way and further described below.

We can also start a thread eagerly if a condition $c$ holds by writing \inlinealg{\,spawn$\colon$}\inlinealg{\,$\{$eager if $c;$}\inlinealg{\,$\cdots\}$\,}.
When the condition holds additionally specified priorities are ignored.

\paragraph{Scoped Thread Groups}

A scoped thread group $T$ is defined within a lexical scope of the program.
We write this as \inlinealg{\,thread group $T\colon$ $\{\ldots\}$\,}.
Within that scope, multiple threads can be spawned into that thread group.
When spawning a thread within thread group $T$, a key $K$ has to be specified: \inlinealg{\,spawn $K$ in $T\colon$ $\{\cdots\}$\,}.

After spawning any number of threads, we can poll for any one of those threads to finish using \inlinealg{\,if $K \gets$ poll $T\colon$ $\{\cdots\}$ else$\colon$ $\{\cdots\}$\,}.
This waits for a thread that hasn't been handled by a previous polling operation to finish and as soon as that happens it assigns the key with which that thread was spawned to $K$ before entering the \textbf{if} block.
If all threads spawned into that group have finished and have been polled (or have been canceled), the poll operation immediately executes the \textbf{else} block instead.

If a poll is performed, before blocking the execution to wait for a thread to finish, the poll operation ensures that at least one of the threads spawned into the polled thread group is running i.e.\ is in the active queue.
If that is not the case, it starts the thread with the smallest \textbf{group priority} among those that that have not been started.

It is also possible to cancel one or more threads identified by their spawning keys $K_1, \ldots, K_n$ using \inlinealg{\,cancel threads $\{K_1, \ldots, K_n\}$ in $T$\,}.
This not only cancels the identified threads, but also all threads spawned recursively from within these threads.

Whenever the scope of the thread group is left all the threads spawned within are canceled, again with recursive cancellation.
This happens either when execution reaches the end of the contained statements or when a control flow statement causes execution to leave the thread group's scope, for example a \textbf{return} leaving the function containing the thread group statement.

\paragraph{Cancelable Scoped Locks}

A scoped lock is held within a lexical scope of the program, indicated by \inlinealg{\,lock $K\colon$ $\{\ldots\}$\,}.
Here $K$ is the locked key.
Only one thread at a time is allowed hold a lock on the same key.
If another thread tries to locks the same key, it is blocked until the lock on that key is available again.

Whenever the scope of the lock is left, the lock is released.
Like thread group cancellation this happens no matter what control flow causes the scope to be left.
In addition to that, if the thread holding that lock is canceled, the lock is also released.

\subsection{Pseudocode}

Now we are ready to present the implemented parallel algorithm as pseudocode.
The code assumes the availability of subroutines computing the well-behaved interior of a pruned sequence set $(X/i)^\circ$, prunable channels $p(X)$, successor sequence sets $\mathcal U(X)$ and canonical sequence sets $X^\sim$.
It also makes use of the threading runtime described above.

The pseudocode also contains further heuristics in the form of the priorities used for ordering the recursive subproblems.
After trying several different variations, I found these to perform well in practice, but a comprehensive comparison of different heuristics would require further experiments.

The core of the algorithm is split into three routines.
The first is the \proc{ParallelImprove} routine which corresponds to \proc{Improve} with \proc{PrunedImproveStep} inline.

The second routine is \proc{ParallelHuffmanImproveStep} which is called from \proc{ParallelImprove} and corresponds to \proc{HuffmanImproveStep} with \proc{HuffmanUpdate} inline.

The third and final routine is \proc{ParallelSuccessorImproveStep}, also called from \proc{ParallelImprove}, and corresponding to \proc{SuccessorImproveStep} with \proc{SuccessorUpdate} inline.

To avoid race conditions, the value $c$, in which \proc{Improve} stored the initial value of $b(X)$ to detect progress, is passed as an argument to \proc{ParallelImprove}.
That way \proc{ParallelImprove} can immediately return when a different thread improved the bounds of $b(X)$ in the time between when \proc{ParallelImprove}'s caller last observed $b(X)$ and when \proc{ParallelImprove} obtains the lock on $X$.
The callers of \proc{ParallelImprove} also keep a local cache $b'(X)$ of the $b(X)$ values they observed, only updated when a polled thread finishes, which helps us avoiding race conditions.

Also note that for improvement via pruning, we omit the loop we had in \proc{PrunedImproveStep}.
Given a sequence set $X$ with a single prunable channel $i$, the bounds $b(X/i)$ are always contained in $b(X)$, as this is the case for the initial $b = b_0$ and we only update $b(X)$ by setting it to $b(X/i)$.

Another difference is that we use a global map $u(X)$ to cache all values $u$ which we computed in \proc{HuffmanUpdate}.
This allows us to quickly skip \proc{ParallelHuffmanImproveStep} without recomputing the Huffman bound, once its potential is exhausted for a given sequence set.

\begin{algo}
	\algorithm{
		global $b \gets b_0 = \cdots$		\cmnt{As before}

		global $u \gets ( X \mapsto 0 )$ \cmnt{Caching \proc{HuffmanUpdate}'s $u$}

		\AW

		procedure $\proc{ParallelImprove}(X,\, c,\, r)\colon$ \AC{
			lock $X$ \cmnt{Allow only one simultaneous thread improving $X$} \AC{
				\cmnt{$\downarrow$ If another thread already improved $X$}

				if $\min b(X) = \max b(X)$ or $b(X) \ne c\colon$ return

				\cmnt{$\downarrow$ Check for a single prunable channel}

				for $Z$ in $X^\ang{\mathbf b}$ if $|p(Y)| = 1\colon$  \AC{
					$Y \gets (Z/i)^{\circ\sim}$ where $i \in p(Y)$

					$\proc{ParallelImprove}(Y,\, c,\, r)$

					$b(X) \gets b(Y)$

					return
				}

				\cmnt{Is it possible to improve $b(X)$ using the Huffman bound?}

				if $u(X) > \min b(X)\colon$ \AC{
					$\proc{ParallelHuffmanImproveStep}(Z,\, c,\, r)$

					if $b(X) \ne c\colon$ return
				}

				$\proc{ParallelSuccessorImproveStep}(Z,\, c,\, r)$
			}
		}
	}

	\algorithm{
		procedure $\proc{ParallelSuccessorImproveStep}(X,\, c,\, r)\colon$ \AC{
			thread group $T\colon$ \AC{
				$b' \gets \emptyset$ \cmnt{Local cache to serially process updates of $b$}

				for $Y$ in $\mathcal U(X)^\sim\colon$ \AC{
					$b'(Y) \gets b(Y)$ \cmnt{Cache current bounds for $Y$}

					spawn $Y$ in $T\colon$ \AC{
						\cmnt{$\downarrow$ Eager if improving $X$'s bound requires improving $Y$'s bound}

						eager if $\min b'(Y) \le 1 + \min b'(X)$

						group priority $|Y|,\, \min b'(Y),\, \max b'(Y)$

						global priority $r, |Y|,\, \min b'(Y)$

						$\proc{ParallelImprove}(Y, b'(Y),\, r)$ \cmnt{Recurse to improve $Y$'s bounds}
					}
				}

				loop$\colon$ \AC{
					\cmnt{$\downarrow$ Recompute the bounds for $X$ using the cached bounds}

					{$b(X) \gets \{1 + \min\, \{\, \min z \mid z \in \im b' \,\}, \ldots, 1 + \min\, \{\, \max z \mid z \in \im b' \,\}\}$}

					\cmnt{$\downarrow$ If we improved the lower bound or found tight bounds}

					if $\min b(X) > \min c$ or $\min b(X) = \max b(X)\colon$ return

 					if $Y \gets$ poll $T\colon$ \cmnt{Wait for a thread in $T$ to finish} \ACC{
						$b'(Y) \gets b(Y)$ \cmnt{Update cached bounds for $Y$}

						if $\min b'(Y) \ne \max b'(Y)\colon$ \cmnt{Respawn a thread if $b(Y)$ not yet singleton} \AC{
							spawn $Y$ in $T\colon \cdots$ \cmnt{Identical to the \textbf{spawn} above} %
						}
					}
					else$\colon$  \cmnt{All threads in T finished} \AC {
						return
					}
				}
			}
		}
	}

	\algorithm{
		procedure $\proc{ParallelHuffmanImproveStep}(X,\, c,\, r)\colon$ \AC{
			thread group $T\colon$ \AC{
				$b' \gets \emptyset$ \cmnt{Local cache to serially process updates of $b$}

				\cmnt{$\downarrow$ Canonicalized prunings for $X^\ang{\mathbf b}$ as multisets with duplicates}

				for $c$ in $\ang{\mathbf b}\colon$ $p_c \gets \{\, (X^c/i)^\sim \mid i \in p(X^c) \,\}_\#$

				for distinct $Y$ in $p_c$ with $c$ in $\ang{\mathbf b}\colon$
						\cmnt{Ignore duplicates for spawning threads} \AC{
					$b'(Y) \gets b(Y)$ \cmnt{Cache current bounds for $Y$}

					spawn $Y$ in $T\colon$ \AC{
						group priority $\min b'(Y),\, |Y|,\, \max b'(Y)$

						global priority $r, |Y|,\, \min b'(Y)$

						$\proc{ParallelImprove}(Y, b'(Y),\, r)$ \cmnt{Recurse to improve $Y$'s bounds}
					}
				}

				loop$\colon$ \AC{
					\cmnt{$\downarrow$ Recompute the Huffman bounds for $X$ using the cached bounds}

					for $(c, f)$ in $\ang{\mathbf b} \times \{\min, \max\} \colon$ $H_{c, f} \gets H_{1{+}\mathrm{max}} \,\{\, f(b'(Y)) \mid Y \in_\# p_c \,\}_\#$

					$u(X) \gets \max_{c \in \ang{\mathbf b}} H_{c, \max}$ \cmnt{Cache the Huffman bounds' upper bound}

					$b(X) \gets b(X) \cap \{ \max_{c \in \ang{\mathbf b}} H_{c, \min}, \ldots \}$ \cmnt{Update with lower bounds}

					\cmnt{$\downarrow$ If we improved the lower bound or found tight bounds}

					if $\min b(X) > \min c$ or $\min b(X) = \max b(X)\colon$ return

					\cmnt{$\downarrow$ If the Huffman bound cannot lead to further improvement}

					if $\min b(X) \ge u(X)\colon$ return

					\cmnt{$\downarrow$ If the Huffman bound for one of $X^\ang{\mathbf b}$ cannot improve our bound}

					for $c$ in $\ang{\mathbf b}$ if $\min b(X) \ge H_{c, \max}\colon$
					cancel threads $p_c \setminus p_{c\mathbf b}$ in $T$

					if $Y \gets$ poll $T\colon$ \cmnt{Wait for a thread in $T$ to finish} \ACC{
						$b'(Y) \gets b(Y)$ \cmnt{Update cached bounds for $Y$}

						if $\min b'(Y) \ne \max b'(Y)\colon$ \cmnt{Respawn a thread if $b(Y)$ not yet singleton} \AC{
							spawn $Y$ in $T\colon \cdots$ \cmnt{Identical to the \textbf{spawn} above}
						}
					}
					else$\colon$  \cmnt{All threads in T finished} \AC {
						return
					}
				}
			}
		}
	}

\end{algo}

\newpage

\end{document}